\documentclass[12pt]{article}
\pdfoutput=1
\usepackage{jheppub}
\usepackage{simplewick}
\usepackage{verbatim}
\usepackage{graphics, color}
\usepackage{graphicx}
\usepackage{amsmath}
\usepackage{amssymb}
\usepackage{amsthm}
\usepackage{amsfonts}
\usepackage[usenames,dvipsnames]{xcolor}
\usepackage[normalem]{ulem}
\usepackage{dsfont}
\usepackage{mathabx}
\usepackage{enumitem}
\usepackage[english]{babel}
\usepackage{slashed}

\newtheorem{theorem}{theorem}[subsection]

\newcommand{\be}{\begin{equation}}
\newcommand{\ee}{\end{equation}}

\newcommand{\lan}{\langle}
\newcommand{\ran}{\rangle}
\newcommand{\Tr}{\mathrm{Tr}}

\newcommand{\mO}{\mathcal{O}}

\newcommand{\wt}{\widetilde}
\newcommand{\wb}{\widebar}

\newcommand{\mP}{{\cal{P}}}
\newcommand{\mE}{{\cal{E}}}
\newcommand{\mR}{{\cal{R}}}
\newcommand{\mI}{\mathbb{I}}

\newcommand{\mH}{\mathcal{H}}

\definecolor{grey}{rgb}{.5,.5,.5}
\definecolor{bluegreen}{rgb}{0,.5,.5}
\definecolor{darkgreen}{rgb}{0,.5,0}

\newcommand{\beq}{\begin{equation}}
\newcommand{\eeq}{\end{equation}}
\def\({\left(}
\def\){\right)}

\begin{document}

\title{Holographic Quantum Error Correction and the Projected Black Hole Interior}
\author{Ahmed Almheiri}
\affiliation{Institute for Advanced Study,  Princeton, NJ 08540, USA}
\emailAdd{almheiri@ias.edu}
\abstract{ The quantum error correction interpretation of AdS/CFT establishes a sense of fluidity to the bulk/boundary dictionary. We show how this property can be utilized to construct a dictionary for operators behind horizons of pure black holes. We demonstrate this within the context of the SYK model with pure black hole microstates obtained via projecting out a single side of the thermofield double (and perturbed versions thereof). Assuming an erasure subsystem code for the duality between the eternal black hole and the thermofield double, this projection results in a rewiring of the dictionary so as to map the interior operators to the remaining boundary in a determinable way. We find this  dictionary to be sensitive to the implemented projection in a manner reminiscent of previous state-dependent constructions of the black hole interior. We also comment on how the fluidity of the dictionary can be used to transfer information between two black holes connected by a wormhole, relating the ideas of entanglement wedge reconstruction and the Hayden-Preskill decoding criterion.\\
\\
\begin{center}
{\it Dedicated to the memory of Joseph Polchinski }
\end{center}
}

\maketitle

\section{Introduction}

The enigmatic nature of the black hole interior has received much attention in recent years due to the conflict between semi-classical expectations of a smooth horizon and the treatment of black holes as quantum systems with a finite density of states \cite{Almheiri:2012rt}. These confusions regarding the interior are fundamentally linked to the problem of information loss \cite{Hawking:1976ra}, and it is generally believed that a solution of the former might inform us on the latter.

In the context of the AdS/CFT conjecture, where these paradoxes become sharpest, these issues manifest themselves in the difficulty of establishing a dictionary between interior operators and CFT observables \cite{Almheiri:2013hfa}. A plethora of proposals have been put forward which try to ensure a smooth horizon for an infalling observer \cite{Papadodimas:2012aq, Papadodimas:2013jku, Papadodimas:2015jra, Verlinde:2012cy, Verlinde:2013uja, Verlinde:2013qya, Maldacena:2013xja, Nomura:2012cx, Nomura:2013gna, Nomura:2014woa}. A more or less common strategy of these proposals is to begin with some pure state black hole in AdS/CFT, along with the boundary dual of the bulk algebra of operators outside but near the horizon and then try and find a corresponding boundary algebra which mimics the semi-classical algebra of operators behind the black hole horizon. The goal is to find such an algebra with the condition that they ensure a smooth horizon in the considered pure state. However, primarily because these conditions are enforced without prior knowledge of the actual physics of the interior they tend to create ambiguities that run afoul of the standard rules of quantum mechanics \cite{Harlow:2014yoa, Marolf:2015dia}.

The new framework for understanding the AdS/CFT dictionary as a Quantum Error Correcting (QEC) code \cite{Almheiri:2014lwa} has not yet been utilized to address these issues. This framework was proposed as a resolution of an apparent inconsistency between subregion-subregion duality (SSD) and the properties of operator algebras in quantum field theories. In particular, SSD seems to indicate the existence of non-trivial bulk operators which commute with all local operators in the CFT on a given time slice, in contradiction with Schur's lemma (or the ``time-slice axiom" in continuum QFT \cite{Streater:1989vi, Haag:1992hx}) that they  must then be proportional to the identity. Viewed through the lens of QEC, this conflict is resolved by interpreting the SSD operator identities as subspace statements holding within some code subspace $\mH_{code}$. For example, one can show that if a logical operator (one that acts within the code subspace) $\widetilde{\mO}$ satisfies
\begin{align}
\mP_{code} \big[ \widetilde{\mO}, X_E \big] \mP_{code} = \mP_{code} \big[ \widetilde{\mO}^\dagger, X_E \big] \mP_{code} = 0
\end{align}
for all operators $X_E$ supported on some subsystem $E$ and where $\mP_{code}$ is the projector on the code subspace $\mH_{code}$, then there exists an operator supported on the complement of $E$, denoted by $\widebar{E}$,  such that
\begin{align}
\widetilde{\mO} \mP_{code} = \mO_{\widebar{E}} \mP_{code}, \ \widetilde{\mO}^\dagger \mP_{code} = \mO_{\widebar{E}}^\dagger \mP_{code}
\end{align}
This describes a version of a QEC usually called Operator Algebra Quantum Error Correction (OAQEC) \cite{beny2007generalization, beny2007quantum}.  This framework has also aided in understanding two other important aspects of the AdS/CFT duality. The first is the so-called entanglement wedge reconstruction proposal which states that the density matrix of a boundary subregion $\widebar{E}$ is sufficient to reconstruct the entire entanglement wedge $W_{\widebar{E}}$, the bulk region composed of the union of all spacelike slices bounded by the Ryu-Takayanagi (RT) or Hubeny-Rangamani-Takayanagi (HRT) surface and the boundary subregion itself \cite{Czech:2012bh, Wall:2012uf, Headrick:2014cta}. The proven statement is that any bulk operator with support within $W_{\wb{E}}$ has a dual boundary operator supported purely on $\wb{E}$ \cite{Jafferis:2015del, Dong:2016eik}. Furthermore as shown in \cite{Harlow:2016vwg}, this framework reproduces  the RT \cite{Ryu:2006bv} (HRT \cite{Hubeny:2007xt}) formula  for computing the von Neumann entropy of the region $\widebar{E}$, along with its associated quantum corrections (bulk EFT entanglement entropy) \cite{Faulkner:2013ana},
\begin{align}
S(\rho_{\widebar{E}}) = {A \over 4 G_N} + S(\rho_{W_{\bar{E}}}).
\end{align}
Given the success of this framework it behooves us to apply it to the context of the black hole interior.

The setting in which we will implement these ideas to the black hole interior will be within the duality between AdS$_2$ gravity and (a subsector of) the SYK model. The SYK model is a system of $N$ Majorana fermions randomly coupled via the $q$-local Hamiltonian \cite{Kitaevtalk, Maldacena:2016hyu}, 
\begin{align}
H = (-1)^{q/2}\sum_{i_1 ... i_q}^N J_{i_1...  i_q} \psi_{i_1} ... \psi_{i_q}
\end{align}
for $q \ll N$. This system has been found to reproduce many  features of gravity in AdS$_2$ including the pattern of conformal symmetry breaking at low energies \cite{Almheiri:2014cka}, as well as saturating the bound on chaos typical of commutators in black hole backgrounds \cite{Maldacena:2015waa}. A particularly interesting and controlled setting in which the reconstruction of the black hole interior can be addressed is the Kourkoulo-Maldacena (KM) construction of pure black hole microstates in the SYK model \cite{Kourkoulou:2017zaj}\footnote{See also  \cite{Krishnan:2017txw} for earlier consideration of microstates in an SYK-like model.} (see also \cite{Goel:2018ubv} for further constructions). The KM construction is as follows. First one defines a set of states $| B_s \ran $ which satisfy
\begin{align}
\left( \psi^{2k-1} - i s_k \psi^{2 k} \right) | B_s \ran = 0 \iff  S_k | B_s \ran = s_k | B_s \ran
\end{align}
where $S_k = 2 i \psi^{2k - 1} \psi^{2k}$ is a spin operator with eigenvalues $s_k = \pm1$. This set of states spans the entire Hilbert space of SYK of dimension $2^{N/2}$. One can then obtain black holes of effective temperature $\beta$ by evolving these states in Euclidean time 
\begin{align}
| B^\beta_s \rangle  = e^{-{\beta \over 2} H} | B_s \rangle
\end{align}
which produces an overcomplete basis of black hole microstates of temperature $\beta$. Within the low energy analysis, the geometry of these black holes looks like that of an eternal black hole except that one boundary is excised by an end-of-the-world brane (EWB) which falls into the black hole. Moreover, these states can be prepared by projecting on the thermofield double (TFD) with the CPT invariant state $|B_s\ran$,
\begin{align}
{}_L\lan B_s |\beta \ran_{LR} = | B_s^\beta \ran_{R}  
\end{align}
where
\begin{align}
| \beta \ran_{LR} = {1 
\over \sqrt{Z_\beta}}\sum_{E} e^{- {\beta \over 2} H} |E \ran_L | E \ran_R
\end{align}
Therefore the dual of acting with the projection operator is the insertion of the EWB which falls into the eternal black hole.

It is this latter construction that we will use to find the dictionary for the interior of the pure black hole microstates. The idea is to begin with the eternal black hole, with or without anti-time-ordered shockwaves in the interior, viewed as an erasure subsystem code of \cite{Harlow:2016vwg} describing the dictionary between the left and right exteriors and their corresponding boundaries, and then to study how this dictionary is modified by the projection on the left boundary. We will study this first using a toy model involving random tensors and then prove some general theorems about when and which interior operators may be reconstructed after such projections. We will ultimately find that a necessary and sufficient condition for the reconstructability of an interior subalgebra is given by
\begin{align}
\mP_{code} \big[ \widetilde{\mO}, P_L^s \big] \mP_{code} = \mP_{code} \big[ \widetilde{\mO}^\dagger, P_L^s \big] \mP_{code} = 0
\end{align}
for an interior operator $\wt{\mO}$ and with left projection $P_L^s = | B_s \ran_L \lan B_s |$, and which guarantees the existence of an operator on the right SYK $O_R^s$ such that
\begin{align}
{}_L \lan B_s| \wt{\mO} | \beta \ran_{LR} = O_R^s \ {}_L\lan B_s|\beta \ran_{LR} = O_R^s | B_s^\beta \ran_{R}  
\end{align}
where the superscript $s$ is there to indicate that this operator depends on the particular projection, $P_L^s$. We will discuss the extent of this state-dependence and draw connections to the previous such proposals for the interior. We will discuss how the main reason that this construction avoids the pitfalls of the previous proposals is that the typicality of the state is not the determining factor to the question of the nature of the horizon.

Finally, we will discuss how the fluidity of the dictionary provides a bulk  mechanism for transferring information between two boundary SYK systems dual to an eternal black hole by means of evaporating one system into the other. The information will be transferred in the sense that a message deposited into one boundary will end up in the entanglement wedge of the other, and whose state can then be read off using entanglement wedge reconstruction. We will see that the protocol is very similar to the situation of an evaporating black hole that has reached the Page time, where further infalling messages can be decoded from the Hawking radiation using the Hayden-Preskill protocol upon allowing the black hole to release a few more Hawking quanta \cite{Hayden:2007cs}.

\section{Pure SYK Black Hole Microstates}
\label{projectedmicrostates}

\subsection{KM Construction of Atypical Microstates}

We begin by reviewing the analysis of KM in constructing the states $|B_s \ran$ dual to pure black hole microstates of effective inverse temperature $\beta$ with an end-of-the-world brane (EWB) capping off the spacetime deep inside the interior \cite{Kourkoulou:2017zaj}. This dual bulk description is deduced from the form of the fermion bilinear correlation functions studied in the low energy limit $1 \ll \beta J \ll N$ and working to leading order in the $1/N$ expansion.

Consider the diagonal correlation functions $\lan B_s^\beta | \psi^i (t_1) \psi^i (t_2) | B_s^\beta \ran $ which can be written in terms of the TFD as
\begin{align}
\lan \beta | \Big[ | B_s \ran_L\lan B_s| \otimes  \psi^i (t_1) \psi^i (t_2)  \Big] | \beta \ran 
\end{align}
where the fermions are operators belonging to the right SYK. Since the projections $ | B_s \ran_L\lan B_s| $ for different $s$ form a complete basis, the sum over $s$ just reproduces the thermal expectation value
\begin{align}
\sum_s \lan \beta | \Big[ | B_s \ran_L\lan B_s| \otimes  \psi^i (t_1) \psi^i (t_2)  \Big] | \beta \ran = \Tr \left[ e^{- \beta H}  \psi^i (t_1) \psi^i (t_2) \right] \label{completeness}
\end{align}
At large $N$, the SYK model has an emergent $\mO(N)$ flavor symmetry, of which a particularly interesting subgroup is the flip group 
\begin{align}
\psi^k \rightarrow (-1)^{k-1} \psi^k
\end{align}
In thinking about the doubled system, we denote the flip group as the one which acts identically on both SYKs. This group implements spin flips and therefore relates the different eigenstates $| B_s \ran$ of the spin operator $S_k = 2 i \psi^{2k -1} \psi^{2 k}$. The TFD state is invariant under this subgroup. In particular, both the Hamiltonian and the maximally entangled state in the energy basis are individually invariant, which becomes manifest when written in the $|B_s \ran$ basis:
\begin{align}
\sum_s |B_s \ran_L \lan B_s | \times \sum_E |E\ran_L |E\ran_R = \sum_s | B_s \ran_L | B_s \ran_R
\end{align}
Therefore, we find that the diagonal correlation functions
\begin{align}
\lan \beta | \Big[ | B_s \ran_L\lan B_s| \otimes  \psi^i (t_1) \psi^i (t_2)  \Big] | \beta \ran 
\end{align}
are invariant under $| B_s \ran \rightarrow | B_{s'} \ran$, and hence
\begin{align}
\lan \beta | \Big[ | B_s \ran_L\lan B_s| \otimes  \psi^i (t_1) \psi^i (t_2)  \Big] | \beta \ran  &= 2^{-N/2}\sum_{s} \lan \beta | \Big[ | B_s \ran_L\lan B_s| \otimes  \psi^i (t_1) \psi^i (t_2)  \Big] | \beta \ran \\
&= 2^{-N/2} \Tr \left[ e^{- \beta H}  \psi^i (t_1) \psi^i (t_2) \right]
\end{align}
Therefore diagonal correlation functions are identical to thermal correlation functions at large $N$. In the low energy limit this attains the conformal form
\begin{align}
\lan B_s^\beta | \psi^i (t_1) \psi^i (t_2) | B_s^\beta \ran \sim {1 \over \left[ {\beta J \over \pi} \sinh {\pi (t_1 - t_2) \over \beta}\right]^{2 \Delta}}
\end{align}
One deduces from this that the bulk geometry is just AdS$_2$. We will see next that the off-diagonal correlators reveal the presence of the EWB.

\begin{figure}[t]
\begin{center}
\includegraphics[height=5cm]{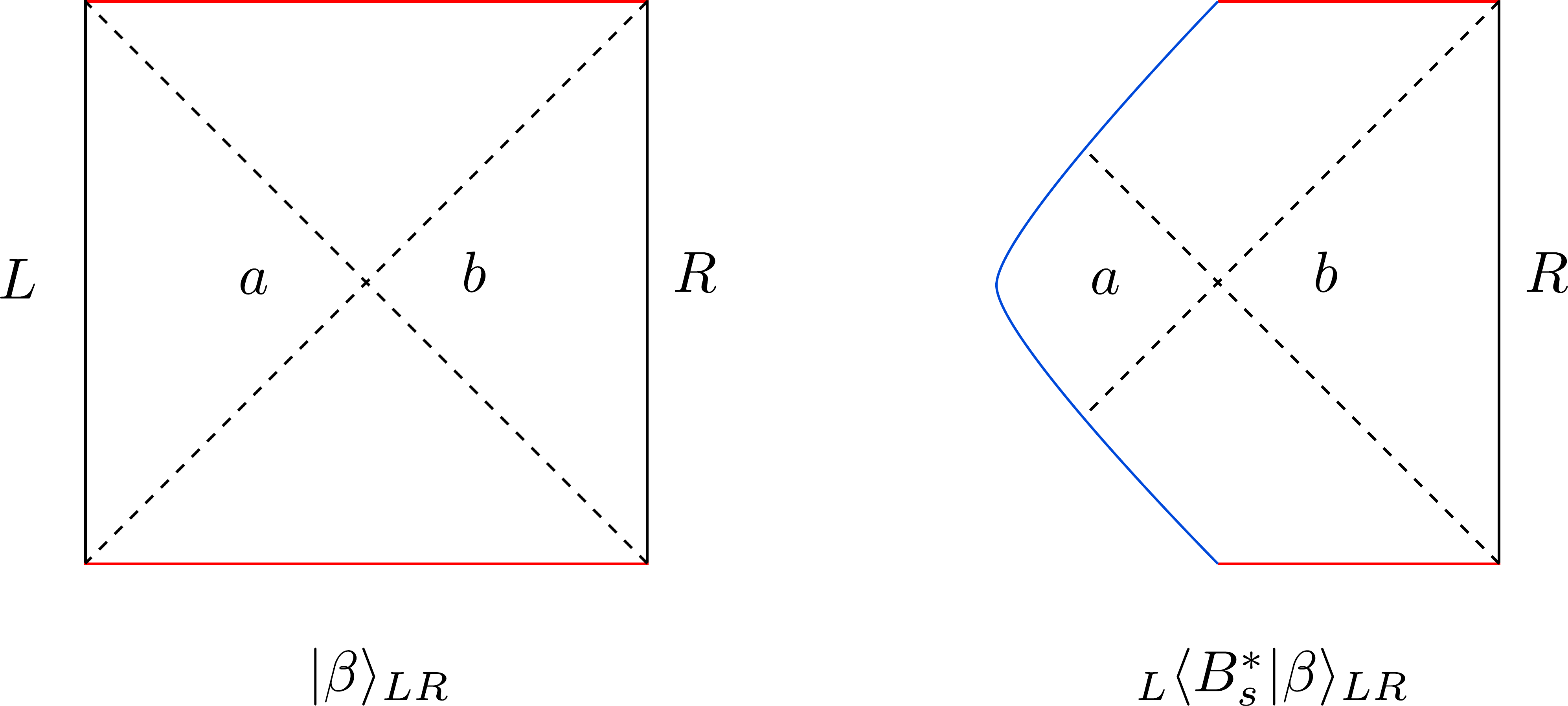}
\caption{The diagram on the left is the standard eternal black hole spacetime dual to the thermofield double state. The projected state on the right is dual to a black hole in a pure state with an end-of-the-world (EWB) brane cutting off the spacetime in the interior. The EWB can be viewed as a UV insertion on the left boundary which then proceeds to fall into the black hole.}\label{sec2projtfd}
\end{center}
\end{figure}

To compute the off-diagonal correlators, say  $\psi^1(t_1) \psi^2(t_2)$, one first notes that while this product is not invariant under the spin group, the following operator is invariant
\begin{align}
S_1 \otimes \psi^1(t_1) \psi^2(t_2)
\end{align}
Moreover, recall that $S_k | B_s \ran = s_k | B_s \ran$. To compute the off-diagonal correlator, one considers
\begin{align}
\lan \beta | \Big[ | B_s \ran_L\lan B_s| \otimes \mI_R \Big] \Big[ S_1 \otimes \psi^1(t_1) \psi^2(t_2) \Big]  \Big[ | B_{\bar{s} }\ran_L\lan B_{\bar{s}}| \otimes \mI_R \Big] | \beta \ran
\end{align}
which simplifies to
\begin{align}
s_1 \lan B_s | \psi^1(t_1) \psi^2(t_2)  | B_s \ran \ \delta_{s \bar{s}}
\end{align}
From the flip group it is clear that the unsimplified correlation function is invariant under the replacement of $B_s \rightarrow B_{s'}$. Again, this invariance means we can sum over the spins $s$ and $\bar{s}$ removing the projectors all together to obtain
\begin{align}
s_1 \lan B_s | \psi^1(t_1) \psi^2(t_2)  | B_s \ran &= 2^{-N/2} \ 2 i \lan \beta | \Big[ \psi^1(0) \psi^2(0) \otimes \psi^1(t_1) \psi^2(t_2) \Big]   | \beta \ran \\
&= 2^{-N/2} \ 2 i \lan \beta |  \psi^1(0) \otimes \psi^1(t_1)     | \beta \ran \lan \beta |   \psi^2(0) \otimes \psi^2(t_2)    | \beta \ran
\end{align}
where the last line is the large $N$ result. This is the product of two left-right diagonal fermion correlation functions, each of which at low energy can be deduced from the single sided correlator by taking a single $t \rightarrow t + i \beta /2$
\begin{align}
\lan \beta |  \psi^1(0) \otimes \psi^1(t_1)     | \beta \ran \sim {1 \over \left[ {\beta J \over \pi} \cosh {\pi t_1  \over \beta}\right]^{2 \Delta}} 
\end{align}
This indicates that the state $|B_s^\beta \ran$ contains the insertion of a high energy operator localized at the single point $\tau  = \beta/2$ on the Euclidean AdS$_2$ boundary.  When continued into Lorentzian time this insertion starts off at $t = 0$ near the left SYK boundary and falls into the black hole. This is shown in figure \ref{sec2projtfd}.

Finally, we review the overcompleteness of this set of states \cite{Kourkoulou:2017zaj}. These black hole microstates are in one to one correspondence with the states $| B_s \ran$, which number at $2^{N/2}$, where $N$ is the number of Majorana fermions in the SYK model. Using the techniques above, we have the (not normalized) overlap
\begin{align}
\lan B_s^\beta | B_s^\beta \ran \equiv  \lan \beta | \Big[ | B_s \ran_L\lan B_s| \otimes  \mI  \Big] | \beta \ran  = 2^{-N/2} Z(\beta)
\end{align} 
Therefore, expanding these states in the energy basis we get
\begin{align}
| B_s^\beta \ran = {1 \over 2^{-N/4} \sqrt{ Z(\beta)}} \sum_{\alpha} e^{- \beta E_\alpha /2} c_\alpha^s | E_\alpha \ran
\end{align}
Note that the sum runs only over half of the energy eigenstates since $(-1)^F = \prod_{k = 1}^{N/2} S_k$ commutes with Hamiltonian; the energy eigenstates which appear in this expression are those that live in the same spin parity sector as $| B_s \ran$. As can be checked numerically \cite{Kourkoulou:2017zaj}, we can assume these coefficients to be random complex numbers which on average satisfy
\begin{align}
c_\alpha^s c_\alpha^{s *} =   2^{- N/2 + 1}    \delta_{s s'}  \label{cav}
\end{align}
This is the expected behavior assuming the states $| B _s \ran$ are  random states in the energy basis. We can compute the overlap of different black hole microstates by assuming the coefficients $c_{\alpha}^s$ to be random unitary matrices\footnote{We thank D. Stanford for discussions on this point.}. With this assumption, we confirm \ref{cav} and also find that on average
\begin{align}
\sqrt{\Big|\lan B_s^\beta | B_{s'}^\beta \ran\Big|^2 } = {\sqrt{2 Z(2 \beta)} \over Z(\beta)} \sqrt{1 - {2^{- N/2} Z^2(\beta) \over Z(2 \beta)}}
\end{align}
which is exponentially small in $N$. This vanishes as $\beta \rightarrow 0$ as required. This non-vanishing overlap is a sign of the `over' in overcompleteness. The `completeness' is shown in equation \ref{completeness}.

\subsection{More Typical Microstates}

The fact that the off-diagonal correlators are not down by powers of $1/N$ is indicative of these microstates being special. We now describe how to prepare more typical microstates where, at the level of two point functions, all but the diagonal correlators are small. These will be black holes with long throats supported by a large number of out-of-time-order (OTO) shockwaves in the interior.

Let $W_L$ represent a left sided unitary operator which creates a series of OTO shockwaves when acting on $| \beta \ran$. In particular, the state
\begin{align}
| W \beta \ran_{LR} \equiv W_L | \beta \ran_{LR}
\end{align}
is dual to the long wormhole supported by OTO shockwaves. Let's assume that the operators $W_L$ are invariant under the diagonal spin group, which implies that $| W \beta \ran_{LR}$ is invariant as well. We want to check that the state 
\begin{align}
| B_s^W \ran_R \equiv {}_L \lan B_s | W \beta \ran
\end{align}
is dual to a single sided black hole with a long throat by computing the diagonal and off-diagonal correlation functions. Just as before, the diagonal correlators can be written as
\begin{align}
{}_R \lan B_s^W | \psi^i(t_1) \psi^i(t_2)  | B_s^W \ran_R  = \lan W \beta | \Big[ | B_s \ran_L\lan B_s| \otimes \psi^i(t_1) \psi^i(t_2)   \Big] | W \beta \ran
\end{align}
and which are invariant under $s \rightarrow s'$. Note, since $W_L$ is unitary, this invariance implies that this expression is equal to 
\begin{align}
\lan \beta | \Big[ | B_s \ran_L\lan B_s| \otimes \psi^i(t_1) \psi^i(t_2)   \Big] |  \beta \ran
\end{align}
which is given by the thermal expectation value as shown before, and so
\begin{align}
{}_R \lan B_s^W | \psi^i(t_1) \psi^i(t_2)  | B_s^W \ran_R \ \propto \ \Tr \left[ e^{\beta H} \psi^i(t_1) \psi^i(t_2)  \right]
\end{align}
The analysis of the off-diagonal  correlator is the same, we have
\begin{align}
s_1 \lan  B^W_s | \psi^1(t_1) \psi^2(t_2)  | B^W_s \ran &= \nonumber \\ 
\lan W \beta | \Big[ | B_s \ran_L\lan B_s| \otimes \mI_R \Big] \Big[ S_1 \otimes  \psi^1(t_1) \psi^2(t_2)  \Big]  & \Big[ | B_{\bar{s}} \ran_L\lan B_{\bar{s}}|  \otimes \mI_R \Big] | W \beta \ran
\end{align}
and which is also invariant under the Flip group. Therefore, by the same arguments above is given by
\begin{align}
& =\lan W \beta | \Big[ 2 i \psi^1(0) \psi^2(0)  \psi^1(t_1) \psi^2(t_2)  \Big]  | W \beta \ran    \\ &  =2 i \lan W \beta |  \psi^1(0) \otimes \psi^1(t_1)   | W \beta \ran \lan W \beta |  \psi^2(0) \otimes \psi^2(t_2)   | W \beta \ran
\end{align}
where the second line is the large $N$ result, and which can be made arbitrarily small as the number of shockwaves is increased. This is the same as what happens in higher dimensions where the left-right correlators die off exponentially in the spatial distance between the boundaries, which here grows arbitrarily with the number of shockwaves  \cite{Shenker:2013yza}. See figure \ref{sec2projlongtfd}.
\begin{figure}[t]
\begin{center}
\includegraphics[height=4cm]{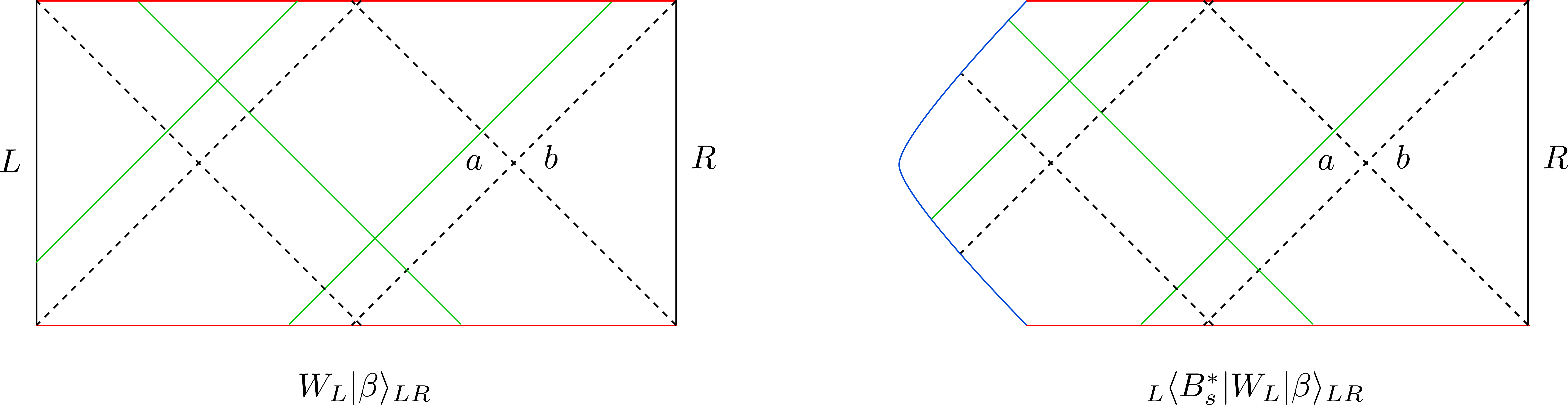}
\caption{Long wormholes supported by out-of-time-ordered (OTO) shockwaves also project into pure states with long throats capped off by an EWB.}\label{sec2projlongtfd}
\end{center}
\end{figure}
Therefore, we have arrived at pure state black holes in SYK that are more reminiscent of typical states in that all simple observables have thermalized, all the while having an understanding of the structure of the interior of the black hole. 


This set of states is also overcomplete and with overlap equal to the atypical case,
\begin{align}
\lan B_s^W | B_{s'}^W \ran  = \lan B_s^\beta | B_{s'}^\beta \ran
\end{align}
It is interesting to compute the overlap between typical and atypical states for the same left projection operator. Using the techniques above we find
\begin{align}
\lan B_s^\beta | B_{s'}^W \ran  = \lan \beta | W | \beta \ran
\end{align}
and therefore for states that differ by a small number of shockwaves (assumed to be created in the same way) the overlap is suppressed by powers of $1/N$. This can potentially be made exponentially small in $N$ by considering a very large number (order $N$) of shockwaves.

\subsection{Bulk Particle Gravitational Dressing and Boundary Energy}

We study in this section the possible ways that a bulk particle maybe be gravitationally dressed in AdS$_2$, and how this dressing affects the trajectory of the boundary and its energy. Our analysis will be completely within the Schwarzian theory coupled to massive bulk matter. We will leave the details to appendix \ref{gravdressing}.

We will work mostly in embedding space and global AdS$_2$ coordinates which are related by
\begin{align}
Y^{-1} &= {\cos t \over \sin \sigma}  \\
Y^{0} &= {\sin t \over \sin \sigma}\\
Y^{1} &=  - {\cos \sigma \over \sin \sigma} 
\end{align}
whose metrics are
\begin{align}
ds^2 &= - \left( d Y^{  -1} \right)^2 -  \left( d Y^{ 0 } \right)^2 +  \left( d Y^{  1} \right)^2, \ Y^2 = -1 \\
ds^2 &= {-dt^2 + d \sigma^2 \over \sin^2 \sigma} 
\end{align}

We begin with the case of a bulk particle in the eternal black hole solution represented by two boundary particles, one for each SYK system. These boundary particles behave as oppositely charged particles in an electric field whose trajectories satisfy \cite{Kourkoulou:2017zaj, Maldacena:2017axo, Maldacena:2016upp}
\begin{align}
Y \cdot Q_{L_\partial} = + q, \ \ Y \cdot Q_{R_\partial} = - q
\end{align}
where $Y^\mu$ is the trajectory in embedding space, and for some $q$. These conditions completely determine the trajectory in terms of the charges.  Ignoring the presence of the bulk particle for the moment, we can pick a gauge where the boundary charges are
\begin{align}
Q^a_{R_\partial} = (Q^{-1}_{R_\partial}, Q^0_{R_\partial}, Q^1_{R_\partial})   = (\sqrt{E}, 0 , 0) = - Q^a_{L_\partial}
\end{align}
The left boundary charge is determined form the right by the gauge constraint condition $Q_{R_\partial} + Q_{L_\partial} = 0$. The energy measured on either boundary is given by the quadratic Casimir constructed from the charges, for e.g.
\begin{align}
H_R = -Q^2_{R_\partial} = E
\end{align}

A massive neutral bulk particle satisfies the condition
\begin{align}
Y \cdot Q_{Bp} = 0
\end{align}
which also completely determines the trajectory. We choose to parameterize this charge as
\begin{align}
Q^a_{Bp} = m ( \sinh \gamma \sin \theta, \sinh \gamma  \cos \theta, - \cosh \gamma)
\end{align}
where $m$ is the mass of the particle, $\gamma$ is the rapidity of the particle relative to global AdS$_2$ coordinates (the world line approaches a null line as $\gamma \rightarrow \infty$), and $\pi - \theta$ is the value of global time $t$ at which the particle passes through the center of AdS$_2$, $\sigma = {\pi/2}$. A typical trajectory is shown in figure \ref{bulkparticletfd}.
\begin{figure}[t]
\begin{center}
\includegraphics[height=6cm]{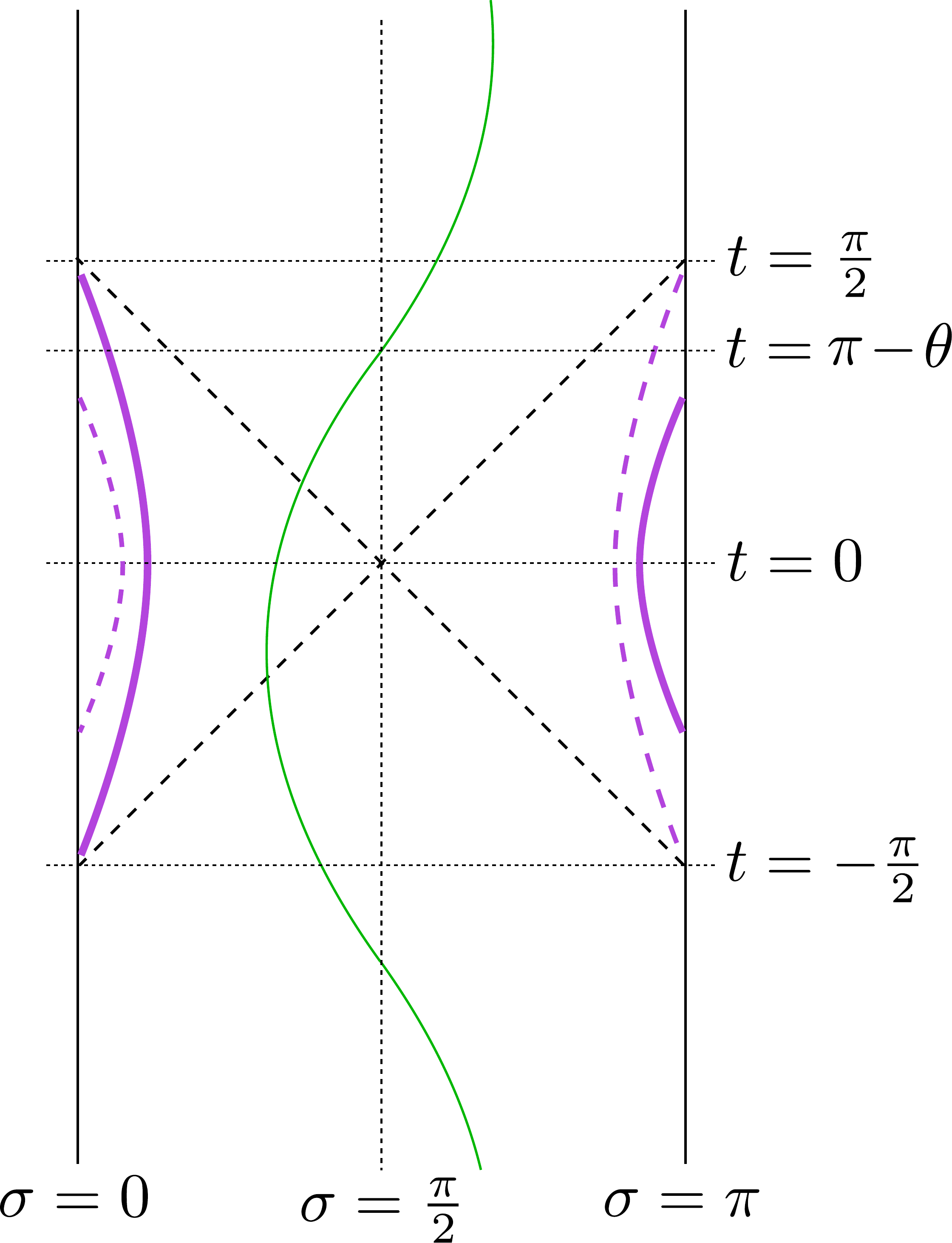} \ \includegraphics[height=6cm]{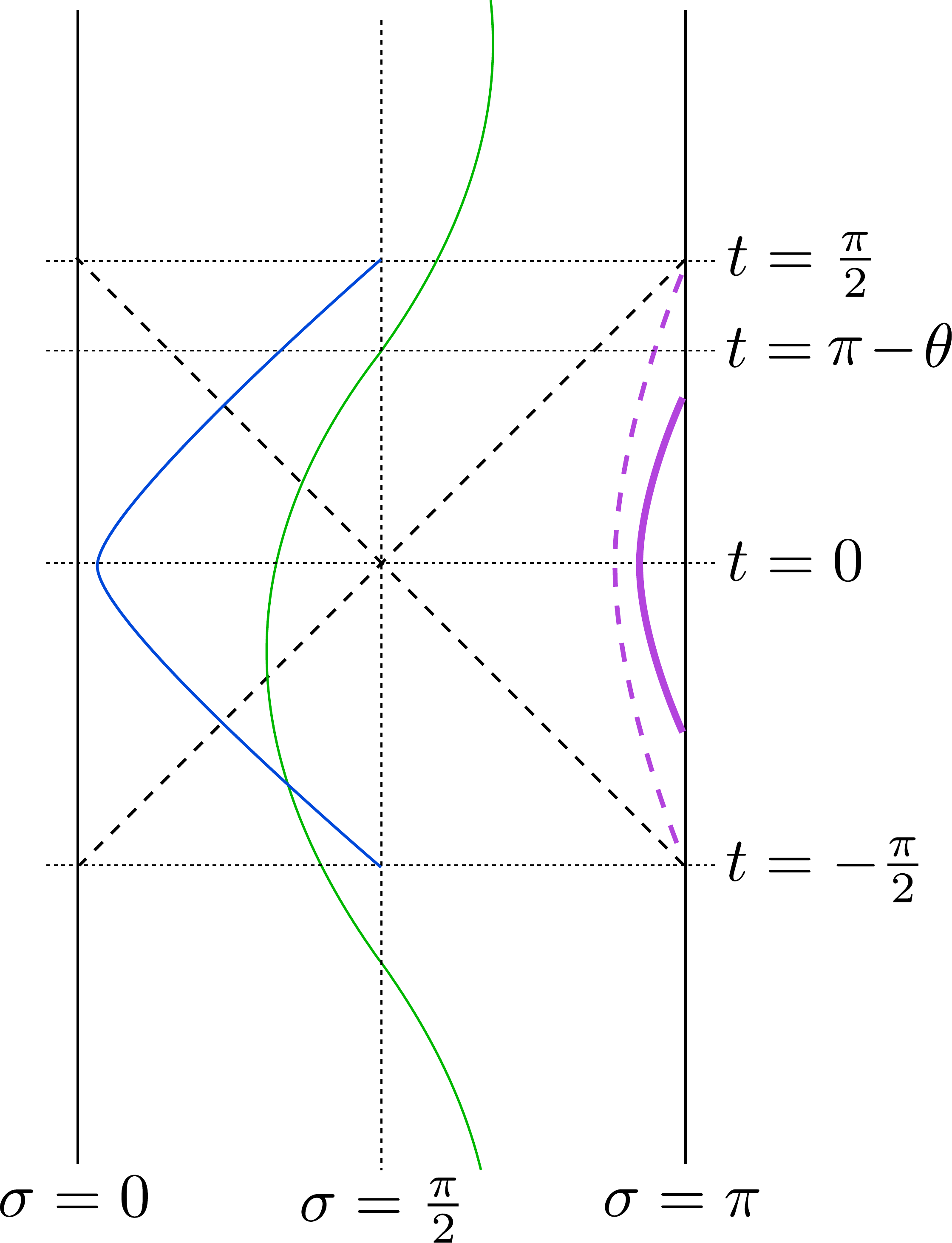}
\caption{Gravitationally dressing a bulk particle (green) to either boundary pushes the boundary trajectories away from the center of AdS$_2$. The solid boundary lines correspond to dressing entirely to the right, and the dotted trajectories are for dressing entirely to the left. The same story holds for the case of the EWB on the right diagram.}\label{bulkparticletfd}
\end{center}
\end{figure}
In order to place this particle inside the eternal black hole spacetime, we must satisfy the new gauge constraint
\begin{align}
 Q^a_{L_\partial} + Q^a_{Bp}   + Q^a_{R_\partial} = 0
\end{align}
We choose to do so by keeping fixed the trajectory of the bulk particle in global coordinates. This amounts to holding fixed $Q_{Bp}$ and modifying the boundary particle charges. This modification is the result of gravitationally dressing the bulk particle to the boundaries, of which there is an infinite number of ways to do so. Two particularly interesting cases is where the bulk particle is either dressed entirely to the left or  entirely to the right. Respectively, this would leave the charge of the right or left unchanged. This has an interesting effect on the energy measured on the boundaries. The energy of the two boundaries when dressing the particle entirely to the right is
\begin{align}
H_L &= E \\
H_R &=  \left( \sqrt{E} - m  \sinh \gamma \sin \theta \right)^2 + \left( m \sinh \gamma  \cos \theta \right)^2 - \left( m \cosh \gamma \right)^2 \\
&= E - 2 m \sqrt{E} \sinh \gamma \sin \theta - m^2 \\
&\approx  E - 2 m \sqrt{E} \sinh \gamma \sin \theta 
\end{align}
where in the last line we took the mass of the particle to be much smaller than mass of the black hole. Since all the dressing is pointing towards the right boundary, we see that the left energy is insensitive to the presence of the bulk particle, as expected. 

The dependence of the right measured energy on the trajectory of the bulk particle is interesting. Let's first consider the case where $\gamma \neq 0$. In the gauge picked, the bulk particle will emerge from the past horizon and fall into the future horizon, and $\theta$ controls on which exterior the particle will emerge into. For $0 < \theta < \pi$ the particle emerges into the left exterior and registers as negative energy on the right boundary, while for $\pi < \theta < 2 \pi$ it emerges into the right exterior and registers as positive energy on the right boundary. For the case where $\theta = \pi n$ for $n \in \mathrm{Z}$ the particle never emerges out of the black hole and it registers as negative energy $m^2$ when dressed to either boundary. This case is identical to the $\gamma = 0$ situation which describes a particle at rest going through the bifurcation point, and in fact the two are related by an SL2 transformation which preserves the boundary particle trajectories, namely Rindler time evolution.

The trajectories of the boundary particles are also modified by the dressing. As discussed in appendix \ref{gravdressing}, the modification is qualitatively the same for all $\theta$, and the boundary trajectory is pushed farther towards the global AdS$_2$ boundary.

The situation with the end-of-the-world brane is nearly identical. The charge of the brane is that of a bulk particle with $\gamma \rightarrow \infty$ and $\theta = \pi/2$. Dressing the bulk particle to the boundary leads to the same conclusions as in the eternal black hole in terms of energy and modification of the trajectory. The new thing here is that we have the option of dressing the bulk particle to the brane. Such particles do not change the energy at the boundary, and can be thought of as operators in the single remaining SYK boundary that commute with the Hamiltonian (obviously not as an operator statement but within some subspace of states). Moreover, we find that the trajectory of the brane is not modified but its mass is always decreased independent of the location of the bulk particle\footnote{ Dressing the bulk particles to the brane requires the theory to contain branes of different masses, which is an assumption we make about the UV theory of the bulk. We thank J. Maldacena for pointing this out. }.

\section{Reconstruction of the Interior via Quantum Error Correction}
\label{section3}
\subsection{A Puzzle}
\label{sec3puzzle}
In the previous section we reviewed how to construct pure black hole microstates with apparently smooth horizons by projecting out one side of the TFD.  We discuss in this subsection an issue this raises from the perspective of bulk reconstruction. In particular, we know from subregion-subregion duality (SSD) that one may write the TFD interpreted as an erasure subsystem code \cite{Harlow:2016vwg} as
\begin{align}
| \beta \ran_{LR} = U_L U_R | \psi \ran_{a b} | \chi \ran_{\bar{a} \bar{b}}
\end{align}
where $\mH_a \otimes \mH_{\bar{a}}$ is a subspace of $\mH_L$ and $\mH_b \otimes \mH_{\bar{b}}$ is  a subspace of  $\mH_R$. The state $| \psi \ran_{ab}$ represents the state of the quantum fields on the fixed eternal black hole background, and the code subspace is spanned by all states obtained by acting on the factor. The empty TFD, or the Hartle-Hawking vacuum, is an element of the code subspace. All $LR$ states in this subspace maintain the same state of $\bar{a} \bar{b}$, which represents the fixed background geometry. The unitaries are the so-called encoding  unitaries which control how the bulk state is embedded in the boundary product Hilbert space. 

\begin{figure}[t]
\begin{center}
\includegraphics[height=4cm]{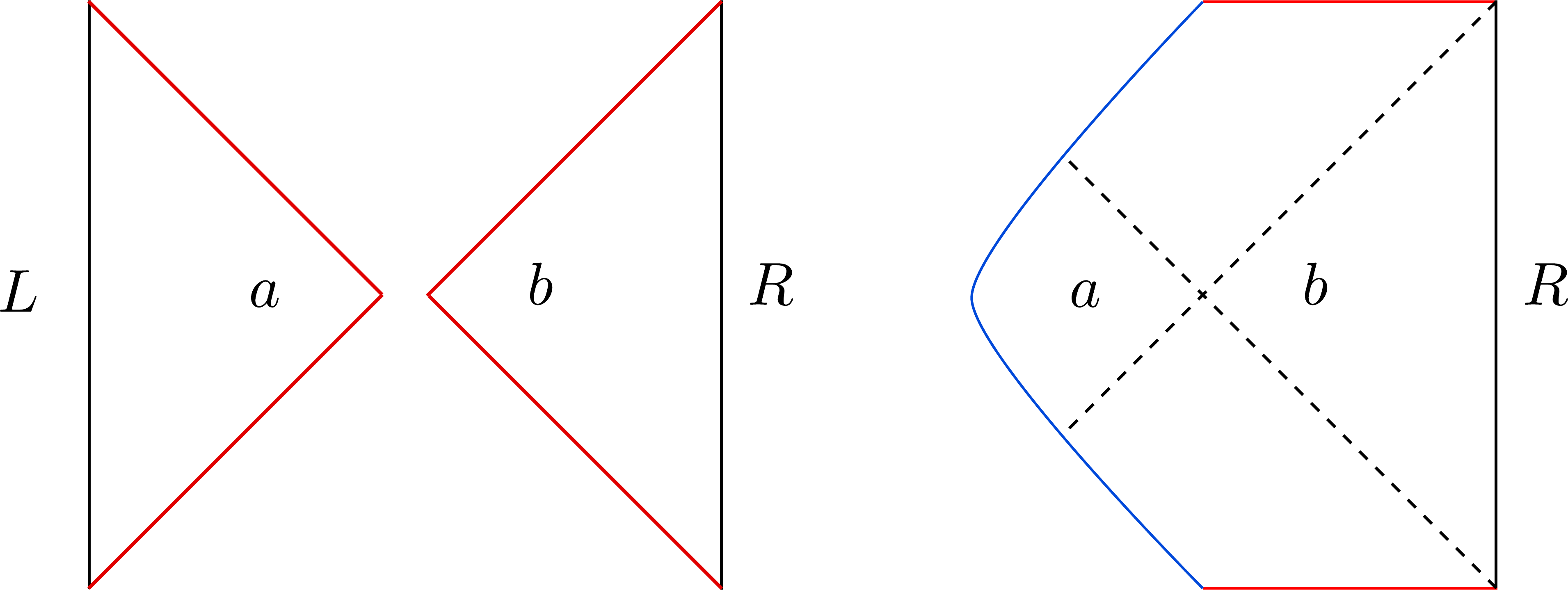}
\caption{The naive expectation (left figure) is that a complete projection on the left boundary would distangle the quantum fields across the horizon forming a firewall. This is inconsistent with the motivated picture from the SYK analysis (right) that this projection generates a pure black hole with an interior and a smooth horizon. }\label{puzzle}
\end{center}
\end{figure}

This code reproduces the RT formula along with the FLM correction \cite{Harlow:2016vwg}. Indeed, the von Neuman entropy of, say, the right boundary is
\begin{align}
S(\rho_R^\beta) = S(\rho_{\bar{a}}^\chi) + S(\rho_{a}^\psi) 
\end{align}
where the subscripts denote the state from which the reduced density matrix is computed. The first term is fixed for all states in this code subspace and can be thought of as the area term $A/4G_N$. The second is the FLM bulk entanglement entropy piece.

Now we can state the puzzle: if we project out the left system in the TFD we will necessarily disentangle $L$ and $R$ and thus naively also  disentangle the two subsystems $a$ and $b$ from each other,
\begin{align}
| \beta \ran_{LR} = U_L U_R | \psi \ran_{a b} | \chi \ran_{\bar{a} \bar{b}} \rightarrow | P \ran_L\lan P | \beta \ran_{LR} \stackrel{?}{=} | P \ran_{a \bar{a}} | P \beta \ran_{b \bar{b}}
\end{align}
Therefore, it would seem that the bulk state will necessarily factorize into an unentangled state of the quantum fields across the horizon! This is a recipe for a firewall. This is inconsistent with the constructions of the previous section where the smoothness of the horizon was maintained after the action of the projection. See figure \ref{puzzle}.

The rest of this paper is about the resolution of this puzzle and its related consequences. We will see that the flaw in the last argument is the assumed rigidity of the AdS/CFT dictionary relating the bulk and boundary Hilbert spaces. We will show how the QEC interpretation of the duality produces a fluid dictionary which maintains the entanglement across the horizon. Moreover, this construction produces an explicit reconstruction map for the operators behind black hole horizon. It will be clear that this dictionary will be `state-dependent' providing a concrete realization of the recent ideas of reconstructing the interiors of black holes \cite{Nomura:2014woa, Nomura:2013gna, Nomura:2012cx, Verlinde:2013uja, Verlinde:2013qya, Verlinde:2012cy, Papadodimas:2015jra, Papadodimas:2013jku, Papadodimas:2012aq}.

\subsection{Toy Model: Projected Random Tensor}

\label{toytensorproj}

We begin by considering a toy model for the AdS/CFT correspondence constructed out of a network of random tensors \cite{Pastawski:2015qua, Hayden:2016cfa}. A random tensor is a quantum circuit which prepares a set of qubits in a random state in, say, the computational basis. Let's define such a tensor that prepares a state in the product Hilbert space $\mH_L \otimes \mH_a \otimes \mH_{H_L}$, with $|\mH_L | \gg  |\mH_a| \times  |\mH_{H_L}|$, and also that $|\mH_a| \ll |\mH_{H_L}|$. The prepared state is
\begin{align}
| T \ran = U_{Rand} | 00...0 \ran_{LaH_L} = \sum_{i k} |\psi_{i k} \ran_L | i \ran_a | k \ran_{H_L} \label{singletensorequation}
\end{align}
where the sum runs over an entire basis of $\mH_a \otimes \mH_{H_L}$. The random unitary $U_{Rand}$, guarantees that ${}_L\lan \psi_{i k} | \psi_{i' k'} \ran_L = \delta_{i i'} \delta_{k k'} $,  implying both that there is no mutual information between $a$ and $H_L$ and that both are maximally entangled with $L$. The subspace of $L$ spanned by $\{ | \psi_{i k} \ran_L\}$ is the code subspace of the $\mH_L$.

\begin{figure}[t]
\begin{center}
\includegraphics[height=4cm]{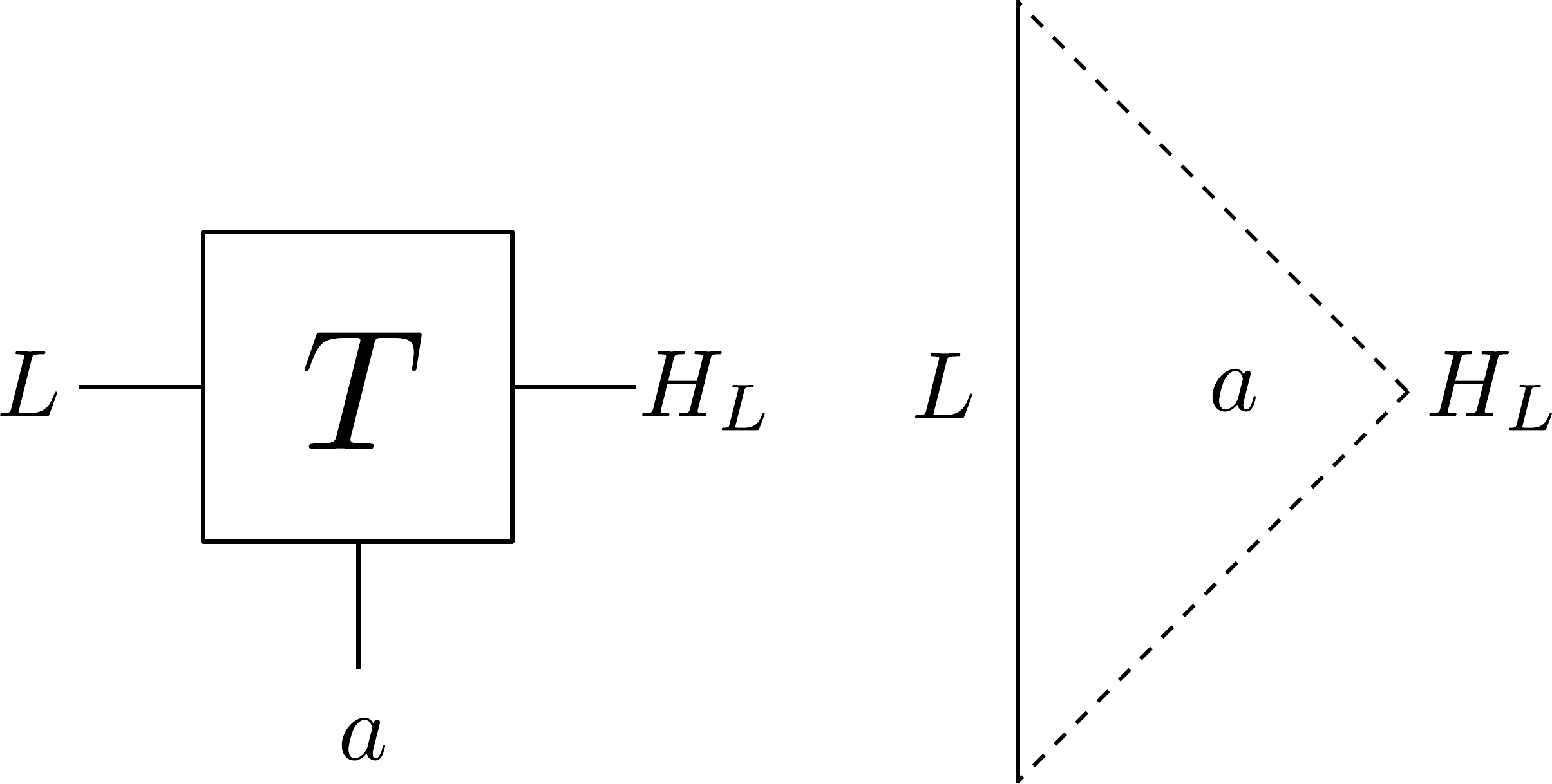}
\caption{A single tensor can be viewed as the encoding of the state of a pure black hole's horizon degrees of freedom $H_L$ and a set of external modes $a$  into the boundary degrees of freedom $L$.  }\label{singletensor}
\end{center}
\end{figure}

This tensor can be thought of as a simplified version of a holographic dictionary for a pure black hole in AdS with boundary $L$, a set of low energy exterior modes $a$, and horizon degrees of freedom $H_L$. The dictionary is implemented in the following way: Given a bulk state $| \phi \ran_{a H_L}$ we can obtain its boundary dual by projecting its complex conjugate on the tensor state as follows
\begin{align}
|\Psi_{\phi} \ran_L = {}_{aH_L} \lan \phi^* | T \ran.
\end{align}
The complex conjugation is just a convenience in order to guarantee that $\sum_{i k} \alpha_{i k} | i \ran_a | k \ran_{H_L}$ maps to $\sum_{i k} \alpha_{i k} | \psi_{i k}\ran_{L}$.

From this we can deduce an operator dictionary.  An operator $\mO_{a H_L}$ on $a H_L$ would be `dual' to an operator $\mO_L$ on $L$ if $\mO_{a H_L} | \phi \ran_{a H_L}$ maps to $\mO_L | \Psi_\phi \ran_L$ for all $| \phi \ran_{a H_L}$. Take for instance a bulk operator supported only on $a$, and an operator on $L$ which satisfies this duality criterion
\begin{align}
\mO_a \otimes \mI_{H_L} = \sum_{i j} \mO_{i j}  |i\ran_{a} \lan j | \otimes \mI_{H_L} \rightarrow \mO_L = \sum_{i j k} \mO_{i j}  | \psi_{i k }\ran_{L} \lan  \psi_{j k}| 
\end{align}

\begin{figure}[t]
\begin{center}
\includegraphics[height=4cm]{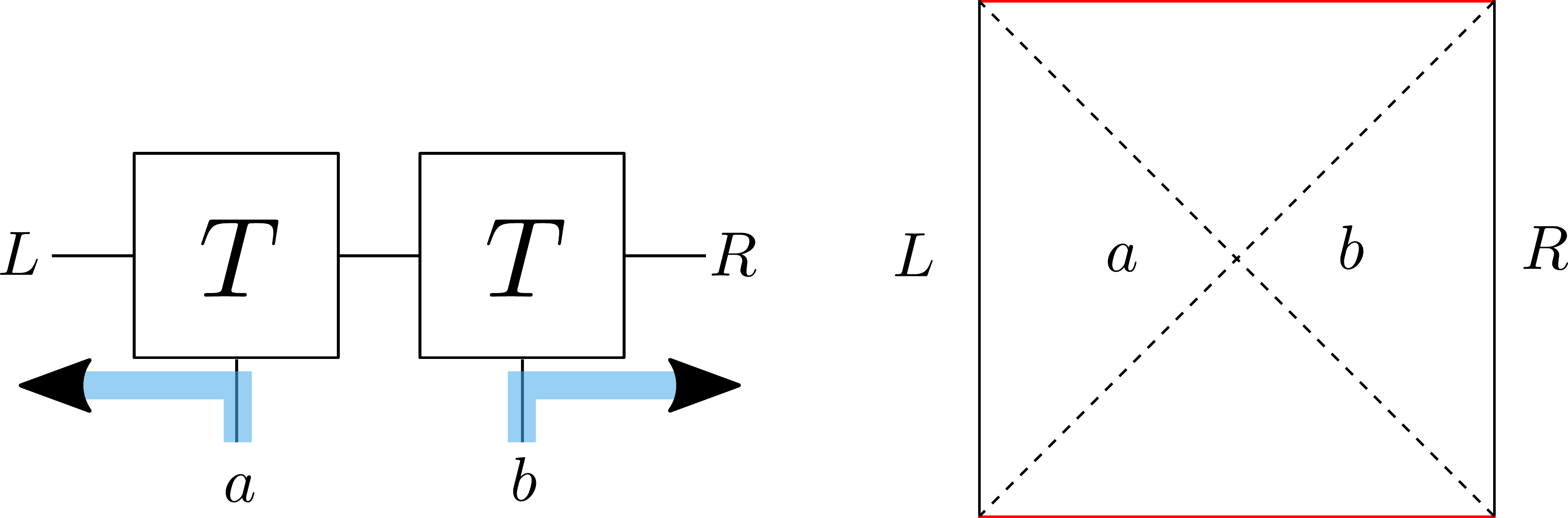}
\caption{The combined tensor produces a subsystem code describing the encoding of the two exterior sets of modes, $a$ and $b$, into their corresponding boundaries, $L$ and $R$. }\label{combinedtensor}
\end{center}
\end{figure}

We can get a toy model for the eternal black hole by sowing two such codes via summing over the horizon indices $H$. We denote the new tensor by $| TT \ran$
\begin{align}
| TT \ran = \sum_{i j k} |\psi_{i k} \ran_L | \psi_{j k} \ran_R | i \ran_{a} | j \ran_{b}
\end{align}
It's not hard to see that this code satisfies subregion-subregion duality, namely
\begin{align}
\mO_a \otimes \mI_{b} \rightarrow \mO_L \otimes \mI_{R}, \ \mI_{a} \otimes\mO_b \rightarrow \mI_{L} \otimes \mO_R
\end{align}
using a similar map to the single tensor case, and furthermore satisfies the quantum corrected RT formula. Take for instance a state $|\phi \ran_{ab}$ which maps to the state $|\Psi_{\phi} \ran_{LR}$. The von Neumann entropy of $R$ in this state is
\begin{align}
S(\rho_R^{\Psi_{\phi} }) = | \mH_H | +  S(\rho_b^{\phi })
\end{align}
where $ | \mH_H |$ comes from summing the index $k$ and can be regarded as reproducing the area term of RT, and $\rho_b^\phi = \Tr_a | \phi \ran \lan \phi |, \ \rho_R^{\Psi_\phi} = \Tr_L | \Psi_\phi \ran \lan \Psi_\phi |$. The second term is the FLM quantum correction to the RT formula.

Now we study how the correction properties of this code get modified by the action of a projection operator on the $L$ system. In particular, we want to see if a state $|\phi \ran_{ab}$ is preserved under the action of a projector on $L$. We check this through the following series of steps:
\begin{enumerate}
\item Project the state $| \phi^* \ran_{ab}$ on the TFD tensor network to obtain the boundary dual of $| \phi \ran_{ab}$
\begin{align}
| \phi \ran_{ab} \rightarrow | \Psi_{\phi} \ran_{LR} = {}_{ab} \lan \phi^* | TT \ran 
\end{align}
\item Act on the left boundary with the projection operator $|P\ran_L \lan P |$ to obtain a new product state of the two boundaries
\begin{align}
| P \ran_L \lan P | \Psi_{\phi} \ran_{LR}
\end{align}
\item Run this new product state through the old tensor network to generate its dual bulk state. The question of interest is: what are the conditions on $P$ such that we regain the original bulk state $| \phi \ran_{ab}$
\begin{align}
{}_{LR}\lan \Psi_{\phi} | P \ran_{L} \lan P | TT \ran \stackrel{?}{=} | \phi \ran_{ab}
\end{align}
Note that a more convenient interpretation of the left hand side of this equality is the projection of a new state of the right boundary ${}_L \lan P | \Psi_{\phi} \ran_{LR}$ on the new projected tensor network ${}_L \lan P | TT \ran$. This new projected tensor network represents the new dictionary post projection.
\end{enumerate}

\begin{figure}[t]
\begin{center}
\includegraphics[height=4cm]{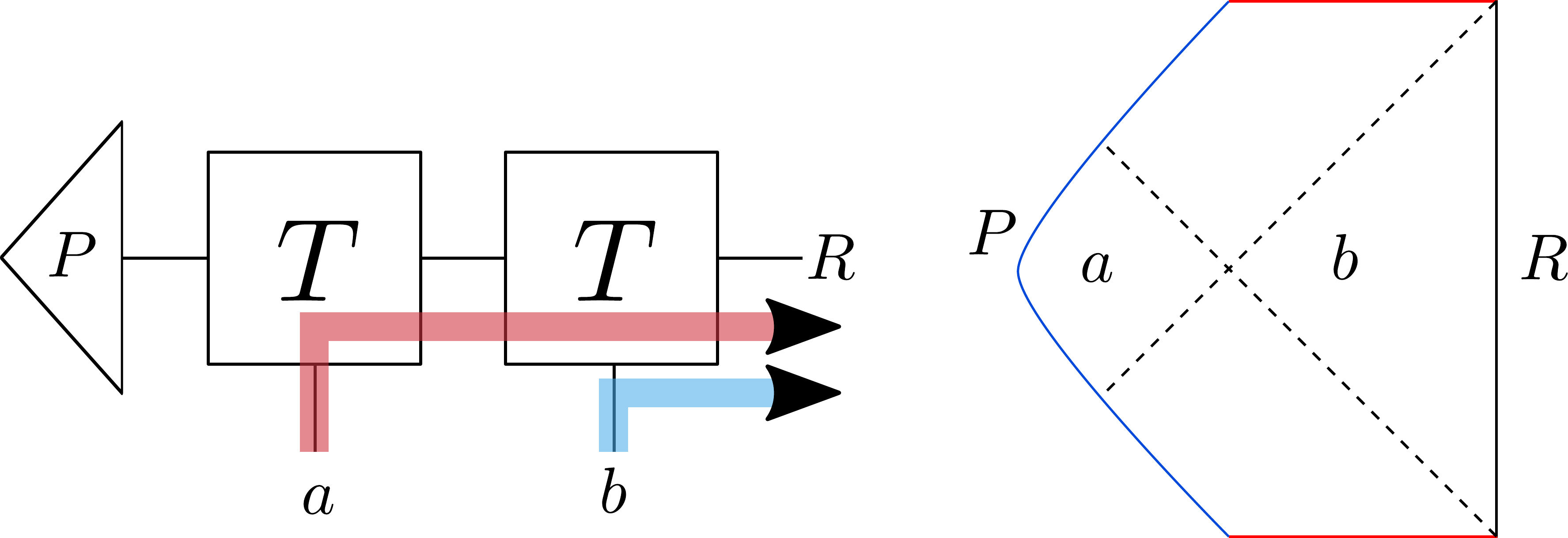}
\caption{The projected tensor now describes a mapping of both $a$ and $b$ into the right boundary $R$. The map from $b$ into $R$ is the same as in the unprojected case and does not depend on $P$, while the map from $a$ into $R$ depends on it sensitively.}\label{projectedtensor}
\end{center}
\end{figure}

For this to be true for all bulk states, it is necessary and sufficient to apply it to a basis
\begin{align}
\sum_k {}_{R}\lan \psi_{j k} | {}_{L}\lan \psi_{i k} | P \ran_{L} \lan P | TT \ran \stackrel{?}{=} | i \ran_{a} | j \ran_{b}
\end{align}
The left hand side simplifies to
\begin{align}
\sum_{i'} \left( \sum_k  {}_L\lan \psi_{i k}| P \ran_L \lan P | \psi_{i' k}\ran_L \right)  | i' \ran_a | j \ran_b
\end{align}
Therefore, to preserve the bulk state we require that
\begin{align}
\sum_k  {}_L\lan \psi_{i k}| P \ran_L \lan P | \psi_{i' k}\ran_L = \delta_{i i'} \label{errorcond315}
\end{align}
This condition is a standard QEC condition on the set of correctable errors, namely that they act as the identity within the code subspace. From the bulk, this says that the insertion of the end-of-the-world brane does not alter the state of the bulk quantum fields. This is not exactly correct, and we'll consider the more realistic situation in section \ref{oap}.

This condition can be satisfied by choosing $P$ such that the projection of $| P \ran_L$ onto the code subspace is a random state in the basis $| \psi_{i k} \ran$; equivalently that $ \lan P | \psi_{i' k}\ran$ are random complex numbers. This satisfies the necessary equality to an accuracy of $ \sqrt{|a| / |H_L|}$.  Note that we can easily pick a projection which does not preserve the bulk state, for example $| P \ran  =\sum_{k} \alpha_k | \psi_{1k} \ran$, for any $\alpha_k$. This will necessarily break the entanglement between the bulk modes creating a firewall.

We can also determine the operator map from the bulk legs $a$ and $b$ into $R$ after the projection. The goal is to find for every logical operator $\mO_{LR}$, dual to some bulk operator $\mO_{ab}$ in the original unprojected tensor code, an operator supported purely on $R$ such that
\begin{align}
{}_L \lan P | \mO_{LR} | \Psi_{\phi} \ran_{LR} = \mO_{R}  \ {}_L \lan P |  \Psi_{\phi} \ran_{LR} \label{opcond316}
\end{align} 
for all $| \phi \ran_{ab}$ where $| \Psi_{\phi} \ran_{LR} =  {}_{ab} \lan \phi^* | TT \ran $. The dictionary for operators on the right exterior is the same as that of the original unprojected tensor giving the map
\begin{align}
\mI_a \otimes \mO_b \equiv \mI_a \otimes \sum_{j j'} \mO_{j j'} |j \ran_b \lan j' | \rightarrow \mO_R^b = \sum_{j j' k} \mO_{j j'} |\psi_{j k} \ran_R \lan \psi_{j' k} |
\end{align}
As for operators originally on the left exterior, which become interior operators after the projection, we have
\begin{align}
 &\mO_a  \otimes  \mI_b \equiv  \sum_{i i'} \mO_{i i'} |i \ran_a \lan i' | \otimes  \mI_b \nonumber \\
& \rightarrow \mO_R^a (P) = \sum_{i i' j k k'} \mO_{i i'}    {}_L \lan P | \psi_{i k} \ran_{L}  |\psi_{j k} \ran_R \lan \psi_{j k'} |    {}_L \lan \psi_{i' k'} | P \ran_{L}
\end{align}
Which can be checked to satisfy \ref{opcond316} assuming \ref{errorcond315}. And therefore we have generated a new dictionary for the interior operators $\mO_R^a(P)$ which looks very different from that of  the exterior operators $\mO_R^b$. A key difference is in the dependence of the interior operators on the projection operator $P$, and therefore on the microstate of the black hole around which our code subspace lives. This dictionary therefore is state-dependent \cite{Nomura:2014woa, Nomura:2013gna, Nomura:2012cx, Verlinde:2013uja, Verlinde:2013qya, Verlinde:2012cy, Papadodimas:2015jra, Papadodimas:2013jku, Papadodimas:2012aq}.

This result shows how the puzzle of the previous subsection is resolved in this model, and that indeed the bulk state and the entanglement across the horizon is maintained. The invalid assumption we made previously was to take the dictionary between the bulk and boundary to be rigid, namely that defined by the state $| TT \ran$. However, what we learn now is that the projected tensor defines a new dictionary generated by acting with the projection operator $_{L}\lan P | TT \ran$. In particular, while prior to the projection the bulk factors $a$ and $b$ were reconstructable in $L$ and $R$ respectively, the post projection tensor network $_{L}\lan P | TT \ran$ maps both to the right boundary $R$. This fluidity of the dictionary is a new observation bound to be critical for general bulk reconstruction.

\begin{figure}[t]
\begin{center}
\includegraphics[height=4cm]{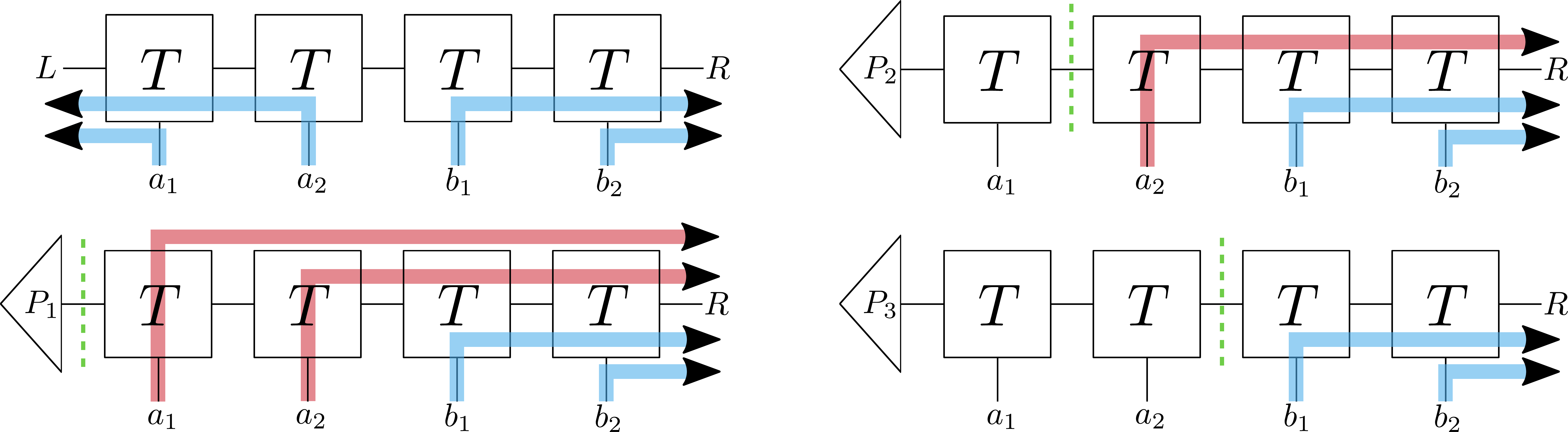}
\caption{This long tensor network can be viewed as either a regular eternal black hole with more external modes in the code subspace, or as a long wormhole where some of bulk legs correspond to modes in the interior. The top left picture represents the standard dictionary for a wormhole where the RT surface in the center. As you go from $P_1$ to $P_3$ the projection is more fine tuned to place the brane, shown in dotted green, at different locations in the bulk.}\label{projectedlongtensor}
\end{center}
\end{figure}

This toy model makes it seem that the entire left exterior is either projected on or remapped to the right, without anything in between. However, this is due to the simplicity of the model having only a single bulk index on each exterior.  We could consider instead combining a number of random tensors that satisfy the subregion subregion duality structure of the thermofield double. Take for instance the case with four bulk legs shown in figure \ref{projectedlongtensor} utilizing four random tensors of different dimensionality. Note that this tensor network can also be thought of as that of a long wormhole where some of the bulk indices correspond to modes in the interior. The tensor state for this network is
\begin{align}
|T^4 \ran = \sum_{i_1 i_2 j_1 j_2 k}  |\psi_{i_1 i_2  k} \ran_L | \psi_{j_1 j_2  k} \ran_R | i_1 \ran_{a_1} | i_2 \ran_{a_2} | j_1 \ran_{b_1} | j_2 \ran_{b_2}
\end{align}
where the sums run over an entire basis of $\mH_{a_1}\otimes \mH_{a_2}\otimes \mH_{b_1}\otimes \mH_{b_2}$, and the states appearing in the factors $L$ and $R$ satisfy $\lan \psi_{i_1 i_2  k} |\psi_{i'_1 i'_2  k'} \ran = \delta_{i_1 i'_1} \delta_{i_2 i'_2} \delta_{k k'}$. For a completely generic projection operator we would reproduce
\begin{align}
\sum_k  {}_L\lan \psi_{i_1 i_2  k}| P \ran_L \lan P | \psi_{i'_1 i'_2  k}\ran_L = \delta_{i_1 i'_1}  \delta_{i_2 i'_2}
\end{align}
However, we could choose a less random projector so that ${}_L\lan \psi_{i_1 i_2  k}| P \ran_L$ are random coefficients without correlations when varying $i_2$ and $k$, but with correlations in the $i_1$ index. This can be chosen to produce, for example, the condition
\begin{align}
\sum_k  {}_L\lan \psi_{i_1 i_2  k}| P \ran_L \lan P | \psi_{i'_1 i'_2  k}\ran_L = \delta_{i_1 1} \delta_{i'_1 1} \delta_{i_2 i'_2}
\end{align}
Therefore, the bulk state would transform after the projection as follows
\begin{align}
|\psi \ran_{a_1 a_2 b_1 b_2} \rightarrow | 1 \ran_{a_1} | \widetilde{\psi} \ran_{a_2 b_1 b_2}
\end{align}
where the latter factor is mapped to the right boundary, as in the top right picture of figure \ref{projectedlongtensor}. The bulk dual of the projection in this case would be a brane which partitions the bulk between the $a_1$ and $a_2$ subsystems\footnote{One can also find a projection which projects on the state of $a_2$ but where $a_1$ is still remapped to the $R$ boundary. It is not obvious what the bulk spacetime would look like for this situation.}.

We've assumed in this section that the size of the bulk Hilbert space $a$ corresponding to the projected black hole interior was smaller than that of the horizon legs $H_L$. This was necessary to ensure the QEC property of establishing a dictionary between the interior and the boundary, and to guarantee that orthogonal bulk states map to orthogonal boundary states. To see how this would fail otherwise, consider again the projected tensor network represented by the state
\begin{align}
{}_L \lan P | TT \ran = \sum_{i j k}  {}_L \lan P | \psi_{i k} \ran_L | \psi_{j k} \ran_R | i \ran_{a} | j \ran_{b}
\end{align}
Consider two (naively) orthogonal bulk states $| i \ran_a | j \ran_a $ and $| i' \ran_a | j' \ran_a $ and compute their overlap after mapping them to the boundary. These states map onto the boundary as 
\begin{align}
| i \ran_a | j \ran_a & \rightarrow \sum_{ k} {}_L \lan P | \psi_{i k} \ran_L | \psi_{j k} \ran_{R} \\
| i' \ran_a | j' \ran_a & \rightarrow \sum_{ k} {}_L \lan P | \psi_{i' k} \ran_L | \psi_{j' k} \ran_{R}
\end{align}
The overlap of these states on the boundary is given by
\begin{align}
\sum_{k k'}     {}_L \lan \psi_{i k'}| P \ran_L   {}_L \lan P | \psi_{i k} \ran_L      {}_R \lan \psi_{j' k'} | \psi_{j k} \ran_{R} =  \delta_{j' j}\sum_{k}     {}_L \lan \psi_{i k'}| P \ran_L   {}_L \lan P | \psi_{i k} \ran_L      
\end{align}
The delta function $\delta_{j' j}$ represents the orthogonality of the exterior bulk states irrespective of the size of the horizon. For the case of $|a| \ll |H_L|$, the second factor should equal $\delta_{i i'}$ giving  the QEC property \ref{errorcond315}. Now, if $|a| > |H_L|$, this property can never be satisfied (a Hilbert space cannot contain a number of mutually orthogonal states greater than its dimension), and therefore the states $| i \ran_a$ and $| i' \ran_a$ are not orthogonal from the boundary perspective. The reason why the size of the horizon is relevant is that it acts as bottleneck that the information of $a$ needs to go through on its way to the right boundary $R$. It would be interesting to study what this means for bulk operators, and whether it implies a departure from their naively expected algebra. In the rest of the paper we will assume that the dimension of the bulk legs is never larger than the dimension of the horizon as to avoid these problems.

\subsection{Projected Quantum Subsystem Correcting Code}
\label{pqscc}
Next we consider the erasure subsystem code of \cite{Harlow:2016vwg} and prove a theorem about how its recovery properties are modified by the projection operator. We will continue with the notation above adapted to the eternal black hole setup.

This code is summarized as follows: Within the two boundary Hilbert space, $\mH = \mH_L \otimes \mH_R$ (assumed to have finite dimension), one can assume the existence of a factorizeable code subspace $\mH_{code} = \mH_a \otimes \mH_b$, whereby $a$ and $b$ correspond to the left and right exteriors respectively. Defining $| \wt{i} \ran$ and $|\wt{j} \ran$ as orthonormal basis states for $\mH_a$ and $\mH_b$,  the following statements, among others, are equiavalent \cite{Harlow:2016vwg}
\begin{enumerate}
\item For $|a| < |L|$ and $|b| < |R|$, the left and right Hilbert spaces can be decomposed as $\mH_L = \left( \mH_{L_a} \otimes \mH_{\wb{L}_a} \right) \oplus  \mH_{\wt{L}} $ and $\mH_R = \left( \mH_{R_b} \otimes \mH_{\wb{R}_b} \right) \oplus \mH_{\wt{R}}$, with $|L_a| = |a|$ and $|R_b| = |b|$ and where $|\wt{L}| < |a|$ and $|\wt{R}| < |b|$. There exists encoding unitary operators $U_L$ and $U_R$ on $L$ and $R$, respectively, such that
\begin{align}
| \wt{i j} \ran_{LR} = U_L  U_R  |i \ran_{L_a} | j \ran_{R_b} | \chi \ran_{\wb{L}_a \wb{R}_b}
\end{align}
for some state $|\chi \ran$ on $\mH_{\wb{L}_a} \otimes \mH_{\wb{R}_b}$.

\item For all logical operators $\wt{\mO}_a$ and $\wt{\mO}_b$ acting within the code subspace $\mH_{code}$, there exist operators $\mO_L$ and $\mO_R$ such that
\begin{align}
\mO_L | \wt{\psi} \ran &= \wt{\mO}_a | \wt{\psi} \ran, \ 
\mO_L^\dagger | \wt{\psi} \ran = \wt{\mO}_a^\dagger | \wt{\psi} \ran \\
\mO_R | \wt{\psi} \ran &= \wt{\mO}_b | \wt{\psi} \ran, \ 
\mO_R^\dagger | \wt{\psi} \ran = \wt{\mO}_b^\dagger | \wt{\psi} \ran 
\end{align}
for any state $|\wt{\psi} \ran \in \mH_{code}$.

\item The reference state
\begin{align}
|\phi \ran = {1 \over \sqrt{|a| |b|}} \sum_{i j} | i \ran_{T_a} | j \ran_{T_b} | \wt{i j} \ran
\end{align}
where $T_a$ and $T_b$ are auxiliary subsystems of dimensions $|a|$ and $|b|$ respectively. The density matrices constructed from this state satisfy
\begin{align}
\rho_{T_a T_b R}(\phi) &= \rho_{T_a}(\phi) \otimes \rho_{T_b R}(\phi) \\
\rho_{T_a T_b L}(\phi) &= \rho_{T_b}(\phi) \otimes \rho_{T_a L}(\phi) 
\end{align}

\end{enumerate}

We refer the reader to \cite{Harlow:2016vwg} for a full proof of the equivalence of these statements. The second point above is the QEC interpretation of subregion subregion duality, and it follows straightforwardly from the first condition. Since $|a| = |L_a|$, there is an isomorphism between operators acting on $\mH_a$ and $\mH_{L_a}$, 
\begin{align}
\wt{\mO}_a | \wt{i} \ran = \sum_{k} \mO_{k i} | \wt{k} \ran \sim \mO_{L_a} | i  \ran_{L_a} = \sum_{k} \mO_{k i} | k \ran_{L_a}
\end{align}
Therefore, an operator on $L$ with the same action can be defined as
\begin{align}
\mO_L = U_L \mO_{L_a}  U_L^\dagger
\end{align}
The same conclusion holds for the right side.

Next, we will study how this code is modified by the action of a projector on the $L$ subsystem. We will see, just as in the random tensor toy model, we can place conditions on the projection operators such that the entire original code subspace continues to be correctable. Consider the following theorem:

\begin{theorem} \label{first}
Consider a subsystem code for the encoding of a code subspace $\mH_{code} = \mH_a \otimes \mH_b$ in a larger physical Hilbert space $\mH_L \otimes \mH_R$ with the properties described above. Consider also a complete projection $P_L \equiv  | P \ran_L\lan P |$ on the subsystem $L$. The following statements are equivalent:
\begin{enumerate}[label=(\roman*)]

\item For $|a| < |\wb{R_b}|$, we consider the decomposition of $\mH_{\wb{R}_b} = \left( \mH_{\wb{R}_b^1} \otimes \mH_{\wb{R}_b^2} \right)  \oplus \mH_{\wb{R}_b^3}$ with $|\wb{R}_b^1| = |a|$ and $|\wb{R}_b^3| < |a|$. The projected code states can be written as
\begin{align}
 { {}_L\lan P | \wt{i j} \ran_{LR} \over \sqrt{N^P}  } = U_R \left( W^P_{\wb{R}_b} \otimes \mI_{R_b}   \right) | i \ran_{\wb{R}_b^1} | j \ran_{R_b} | \wt{\chi} \ran_{\wb{R}_b^2} 
\end{align}
where $U_R$ is the same unitary of the original subsystem code, and $ W^P_{\wb{R}_b}$ is a unitary which depends on the projection, and for some state $| \wt{\chi} \ran_{\wb{R}_b^2} $ and normalization $N^P$.

\item For any logical operator $\wt{\mO}$ of the unprojected code, there exists an operator $\mO_R$ such that
\begin{align}
&{}_L \lan P | \wt{\mO} |\wt{\psi} \ran_{LR} = \mO_R \ \   {}_L \lan P | \wt{\psi} \ran_{LR} \\
&{}_L \lan P | \wt{\mO}^\dagger |\wt{\psi} \ran_{LR} = \mO_R^\dagger  \ \  {}_L \lan P | \wt{\psi} \ran_{LR} 
\end{align}
for any state $|\wt{\psi} \ran$ of the original code subspace.

\item The projection onto the code subspace of  $P_L$ acts identically on the code subspace
\begin{align}
\mP_{code} \  P_L  \ \mP_{code} \  = N^P \  \mP_{code}
\end{align}
where $\mP_{code}$ is the projector on the original $\mH_{code}$, for some positive real number $N^P$. 

\item The projected reference state
\begin{align}
{}_L \lan P | \phi \ran =  {1 \over \sqrt{|a| |b|}}  \sum_{i j} | i j \ran_{T_a T_b} \ { {}_L\lan P | \wt{i j} \ran_{LR}   }
\end{align}
 when normalized, satisfies $I(T_a, T_b) = 0$ and $S_{ent}(\rho_{T_a T_b}) = \ln |\mH_{code}|$.

\end{enumerate}

\end{theorem}

\begin{proof}

\ \\
\begin{itemize}

\item \emph{(i) $\implies$ (ii)}: 

For any logical operator $\wt{\mO}$ we can define an operator
\begin{align}
\mO_R = U_R \left( W^P_{\wb{R}_b} \otimes \mI_{R_b}   \right) \mO_{\wb{R}^1_b R_b}  \left( \left( W^P_{\wb{R}_b}\right)^\dagger \otimes \mI_{R_b}   \right) U_R^\dagger
\end{align}
where $ \mO_{\wb{R}^1_b R_b} $ has support only on $\wb{R}^1_b R_b$ and has the same matrix elements as $\wt{\mO}$. This immediately implies the second property.

\item \emph{(ii) $\implies$ (iii)}:

For all logical operators $\wt{\mO}$ we have
\begin{align}
\wt{\mO} \mP_{code} | P\ran_L \lan P | \mP_{code} &=   \mP_{code} \wt{\mO} | P\ran_L \lan P | \mP_{code} \\
&=   \mP_{code} | P\ran_L  \mO_R \  {}_L\lan P | \mP_{code} \\
&=   \mP_{code}  | P\ran_L \lan P | \wt{\mO} \mP_{code} \\
&=   \mP_{code}  | P\ran_L \lan P | \ \mP_{code}  \wt{\mO}  
\end{align}
We used the property $\emph{(ii)}$ twice. Therefore $\big[ \wt{\mO},  \mP_{code}  P_L \mP_{code} \big] = 0$ for all operators acting with the code subspace. Schur's lemma then guarantees  
\begin{align}
\mP_{code} \  P_L  \ \mP_{code} \  \propto \  \mP_{code}
\end{align}
The left hand side being a positive operator determines the proportionality constant to be a positive really number we can call $N^P$.

\item \emph{(iii) $\implies$ (iv)}:

By direct computation we have
\begin{align}
\rho_{T_a T_b} &=  {1 \over |a| |b|}  \sum_{i i'   j  j'} | i j \ran_{T_a T_b} \lan i' j' |    { {}_{LR} \lan i' j' | P \ran_L \lan P | i j \ran_{LR} \over N^P  } \\
&=  {1 \over |a| |b|} \sum_{i i'   j  j'} | i j \ran_{T_a T_b} \lan i' j' |    \delta_{i i'} \delta_{j j'} \\
&=  {1 \over |a| |b|} \sum_{i    j  } | i j \ran_{T_a T_b} \lan i j | \\
&= \rho_{T_a} \otimes \rho_{T_b}
\end{align}
where in the first step we used $ \mP_{code} \  P_L  \  \mP_{code} \  = N^P \   \mP_{code}$. This factorized density matrix ensures that the mutual information between $T_a$ and $T_b$ vanishes. Moreover, the total density matrix is maximally mixed with dimension $|\mH_{code}|$, and therefore
\begin{align}
S_{ent}(\rho_{T_a T_b}) = \ln |\mH_{code}|
\end{align}

\item \emph{(iv) $\implies$ (i)}:

Inherited from the original unprojected code we have that
\begin{align}
{ {}_L \lan P | \wt{i j} \ran_{LR}} =  U_R   | j \ran_{R_b} \  {}_L \lan P |  U_L |i \ran_{L_a} | \chi \ran_{\wb{L}_a \wb{R}_b}
\end{align}
and therefore  we need to show 
\begin{align}
 {}_L \lan P |  U_L |i \ran_{L_a} | \chi \ran_{\wb{L}_a \wb{R}_b} = \sqrt{N^P}   W^P_{\wb{R}_b}  | i \ran_{\wb{R}_b^1} | \wt{\chi} \ran_{\wb{R}_b^2} 
\end{align}
For some unitary $W^P_{\wb{R}_b}$ on $\wb{R}_b$ and numerical factor $\sqrt{N^P}  $.  This can only be true if $|L_a|=|a| < |\wb{R}_b|$. Dividing $|\wb{R}_b|$ by $|a|$ we get $|\wb{R}_b^2|$ with remainder $|\wb{R}_b^3|<|a|$. Therefore we can consider the Hilbert space factorization $\mH_{\wb{R}_b} = \left( \mH_{\wb{R}_b^1} \otimes \mH_{\wb{R}_b^2} \right)  \oplus \mH_{\wb{R}_b^3}$, with $|\wb{R}_b^1| = |a|$.

The state ${}_L \lan P | \phi \ran$ is a purification of the maximally mixed density matrix of the $T_a T_b$ subsystem, which after normalization must be of the form
\begin{align}
{ {}_L \lan P | \phi \ran \over \sqrt{ N^P_\phi}}  &= \sum_{i j} | i j \ran_{T_a T_b}  V_R  | i j\ran_{R} \\
&= \sum_{i j} | i j \ran_{T_a T_b}  V_R  | j \ran_{R_b}  | i \ran_{\wb{R}_b^1}      | \wt{\chi} \ran_{\wb{R}_b^2}  
\end{align}
For some $V_R$ to be determined based on the Hilbert space factorization of $R$. Requiring this to be equal to the projected reference state we must have
\begin{align}
U_R   | j \ran_{R_b} \  {}_L \lan P |  U_L |i \ran_{L_a} | \chi \ran_{\wb{L}_a \wb{R}_b} =  \sqrt{N^P_\phi} V_R  | j \ran_{R_b}  | i \ran_{\wb{R}_b^1}      | \wt{\chi} \ran_{\wb{R}_b^2}  
\end{align}
for all $ j$, which forces $V_R$ to satisfy
\begin{align}
U_R^\dagger V_R =  W^P_{\wb{R}_b} \otimes \mI_{R_b}  
\end{align}
for some $W^P_{\wb{R}_b}$. It also follows that we should identify the constants $N^P_\phi = N^P$.

\end{itemize}

\end{proof}

\subsection{Operator Algebra Quantum Error Correction from Projected Subsystem Codes}
\label{oap}

The projected subsystem code of the previous section is not quite realized by the KM construction of projecting out one side of the SYK thermofield double. The issue is that the projections considered in KM act nontrivially within the code subspace. Indeed, correlation functions of simple operators receive a modification, for example the off-diagonal fermion correlation functions with and without the projection are:
\begin{align}
&\lan \beta | \mI_L \otimes \psi^1_R(t_1) \psi^2_R(t_2) | \beta \ran \sim \mO(1/N^q) \\
&\lan \beta | P_L \otimes \psi^1_R(t_1) \psi^2_R(t_2) | \beta \ran \sim G_\beta(t_1, i\beta/2) G_\beta(t_2, i\beta/2)
\end{align} 
Therefore, the projection on the code subspace of $P_L$  does not act identically within the code subspace
\begin{align}
 \mP_{code} \  P_L  \  \mP_{code} \  \slashed{\propto}  \   \mP_{code}
\end{align}
This immediately precludes the complete recovery of the state of the code subspace prior to the projection. However, as we will see, it still allows us to recover a subalgebra of logical operators, namely all operators which satisfy
\begin{align}
\big[  \mP_{code} \  P_L  \  \mP_{code}, \wt{\mO}   \big] = 0
\end{align}
This kind of QEC has appeared before in \cite{beny2007generalization, beny2007quantum} and is called Operator Algebra Quantum Error Correction (OAQEC), and was utilized in \cite{Dong:2016eik, Harlow:2016vwg, Almheiri:2014lwa}. 

This condition is motivated from the bulk picture of the brane on the $t = 0$ slice being localized near the boundary and would therefore commute with spacelike separated operators on that same slice. It would  then apply to any bulk operator, inside or outside the horizon, that is dressed to the remaining boundary. The situation is not so clear for the left dressed operators, as those naively do not commute with the projection operator. We conjecture that, in some sense, the part of the operator that extends past the location of the brane into the bulk does commute with the projection, but we fully acknowledge the difficulty of squaring this with bulk diffeomorphism invariance.

Moving on, we will prove the following theorem:
\begin{theorem} \label{them341}
Consider a subsystem code for the encoding of a code subspace $\mH_{code} = \mH_a \otimes \mH_b$ in a larger physical Hilbert space $\mH_L \otimes \mH_R$ with the properties described above. Consider also a  complete projection $P_L \equiv  | P \ran_L\lan P |$ on the subsystem $L$. For any logical operator $\wt{\mO}$ the following statements are equivalent:
\begin{enumerate}[label=(\roman*)]

\item There exists an operator $O_R$, and its Hermitian conjugate $O_R^\dagger$,  with support on $R$ such that
\begin{align}
{}_L \lan P | \wt{\mO} | \tilde{\psi} \ran_{LR} &= O_R  \ {}_L \lan P | \tilde{\psi} \ran_{LR} \\
{}_L \lan P | \wt{\mO}^\dagger | \tilde{\psi} \ran_{LR} &= O_R^\dagger  \ {}_L \lan P | \tilde{\psi} \ran_{LR}
\end{align}
for all states $|\wt{\psi} \ran_{LR} \in \mH_{code}$.

\item The logical operator $\wt{\mO}$ commutes with the projection on the code subspace of the projection operator $P_L$
\begin{align}
\big[  \mP_{code} \  P_L  \  \mP_{code}, \wt{\mO}   \big] &= 0 
\end{align}

\end{enumerate}

\end{theorem}

\begin{proof}
\ \\

\begin{itemize}

\item \emph{(i) $\implies$ (ii)}:

This is identical to the proof of the \emph{(ii) $\implies$ (iii)} implication of theorem \ref{first}.

\item \emph{(ii) $\implies$ (i)}:
On the original subsystem code reference state
\begin{align}
| \phi \ran =  {1\over \sqrt{ |a| |b|}}  \sum_{i j} | i j \ran_{T_a T_b} \ {| \wt{i j} \ran_{LR}  }
\end{align}
we have 
\begin{align}
\wt{\mO} | \phi \ran = \mO^T_{T_a T_b} | \phi \ran
\end{align}
where $\mO^T$ is the transpose of $\wt{\mO}$ but with support on $T_a T_b$. Similarly for the Hermitian conjugate $\wt{\mO}^\dagger$ and $\left(\mO^\dagger_{T_a T_b}\right)^T$.  Notice that the projection of $P_L$ on the code subspace 
\begin{align}
 \mP_{code} \  P_L  \  \mP_{code} = \sum_{i i' j j'} \left( \lan \wt{i j} | P \ran_L\lan \ P | \wt{i' j'} \ran   \right) | \wt{i j} \ran \lan \wt{i' j'} |
\end{align}
has the same matrix elements as the transpose of the reference $T_a T_b$ density matrix of the normalized state
\begin{align}
{{}_L\lan P |  \phi \ran \over \sqrt{ N^P_\phi}} =  {1\over \sqrt{N^P_\phi |a| |b|}}  \sum_{i j} | i j \ran_{T_a T_b} \ { {}_L\lan P | \wt{i j} \ran_{LR}  }
\end{align}
given by
\begin{align}
\rho_{T_a T_b} =  {1\over |a| |b|} \sum_{i i' j j'} \left( {  \lan \wt{i' j'} | P \ran_L \lan P | \wt{i j} \ran   \over N^P_\phi} \right)     | i j \ran_{T_a T_b} \lan i' j' |
\end{align}
This shows the equivalence of  
\begin{align}
\big[  \mP_{code} \  P_L  \  \mP_{code}, \wt{\mO}   \big] = 0 &\Longleftrightarrow \big[ \rho_{T_a T_b}, \mO^T_{T_a T_b}   \big] = 0
\end{align}
and similarly for $\wt{\mO}^\dagger$ and $\left(\mO^\dagger_{T_a T_b}\right)^T$ since the density matrices and projectors are Hermitian.

The next step is to show that this implies the existence of $O_R$ such that
\begin{align}
\mO^T_{T_a T_b}  \  {}_L\lan P |  \phi \ran &= O_R \  {}_L\lan P |  \phi \ran \\
\left( \mO^\dagger_{T_a T_b} \right)^T  \  {}_L\lan P |  \phi \ran &= O_R^\dagger \  {}_L\lan P |  \phi \ran
\end{align}
This has already been proven in \cite{Almheiri:2014lwa}, but we reiterate it here for completeness. We show this by constructing such an $O_R$ and $O_R^\dagger$, and show that they are indeed Hermitian conjugates. For the sake of notational simplicity, we redefine $| I \ran \equiv | i j \ran$. Moreover, we work in a different basis for the $R$ subsystem such that 
\begin{align}
 {}_L \lan P | \phi \ran^N \equiv {{}_L \lan P | \phi \ran \over \sqrt{N^P_\phi}} = \sum_{I K} \alpha_{K I} | I \ran_{T_a T_b} | K \ran_{R}
\end{align}
where $\alpha_{K I }$ can be thought of as a $  |\mH_{code}| \times |\mH_R| $ rectangular matrix. The density matrix of the reference subsystem is given by $\rho_{T_a T_b} = \alpha \alpha^\dagger$. The commutativity of $\mO^T_{T_a T_b}$ with this density matrix ensures that it preserves the subspace of support of $\rho_{T_a T_b}$ on $\mH_{code}$.  Within this subspace $\alpha$ has a right inverse $\alpha^{-1} =  \alpha^\dagger \rho_{T_a T_b}^{-1}$. This allows us to construct $O_R$ as follows:
\begin{align}
\mO^T_{T_a T_b}  \  {}_L\lan P |  \phi \ran^N &= \sum_{I J K}  \left( \mO^T \right)_{J I}  \alpha_{K I}  | J \ran_{T_a T_b} | K \ran_{R} \\
&= \sum_{J K} \alpha_{M J}  | J \ran_{T_a T_b}     \sum_{I L M} \alpha^{-1}_{M L}    \left( \mO^T \right)_{L I}   \alpha_{K I}   | K \ran_{R} \\
&= \sum_{J K} \alpha_{M J}  | J \ran_{T_a T_b}     \left(    \alpha^{-1} \mO^T  \alpha \right)^T | M \ran_{R} \\
&=  \left(  \alpha^T    \mO  ( \alpha^{-1})^T  \right)_{R} \  {}_L\lan P |  \phi \ran^N
\end{align}
and similarly for the Hermitian conjugate $\left( \mO_{T_a T_b}^\dagger \right)^T$. Therefore we have
\begin{align}
O_R &=    \alpha^T    \mO  ( \alpha^{-1})^T \\
O_R^\dagger &=    \alpha^T    \mO^\dagger  ( \alpha^{-1})^T
\end{align}
All we have left to show is that the right hand sides of these expressions truly are Hermitian conjugates of one another. This is easy to see as follows. Starting with the formula for $O_R^\dagger$ and taking the conjugate we get
\begin{align}
(O_R^\dagger)^\dagger &= \left( \alpha^T    \mO^\dagger  ( \alpha^{-1})^T  \right)^\dagger \\
&=   ( \alpha^{-1})^* \mO \alpha^* \\
&= ( \alpha^{-1})^* \mO \alpha^* \alpha^T (\alpha^{-1})^T\\
&= ( \alpha^{-1})^* \alpha^* \alpha^T \mO  (\alpha^{-1})^T\\
&=  \alpha^T \mO  (\alpha^{-1})^T\\
&= O_R
\end{align}
where we used $[ \mO,\alpha^* \alpha^T ] = [ \mO,\rho^T_{T_a T_b} ] =0$ in going between the third and fourth steps.

\end{itemize}

\end{proof}

\subsection{Reconstruction as Teleportation or Active Quantum Error Correction}

The QEC codes used to describe subregion-subregion duality in \cite{ Almheiri:2014lwa, Dong:2016eik, Harlow:2016vwg} belong to the broad class of Erasure codes. These codes are \emph{passive} QEC codes in that they do not involve an error diagnostic step after which a suitable recovery operation is implemented. It is assumed in these codes that one has prior knowledge of which subsystem is going to be corrupted and only then can the information about the code subspace (or a subalgebra) be recovered from its complement. This is naturally suited for the question of subregion-subregion duality in AdS/CFT.

The codes studied in this paper involve a recovery procedure which depends crucially on the details of the projection, $P_L$, and must involve an \emph{active} diagnostic step in order to determine which $P_L$ was acted with on the $L$ subsystem\footnote{We are grateful for discussions on this point with D. Poulin who  demanded a more interesting example of QEC in AdS/CFT beyond passive erasure codes, and hope to have demonstrated such an example in this work.}. There are two equivalent ways of phrasing the recovery procedure: Either as quantum teleportation where knowledge of a measurement result on the entangled $LR$ system in some basis $P^k_L$ informs the correct teleportation protocol, or as an active QEC involving a diagnostic step on the already projected $L$ subsystem to determine which $P^k_L$ was acted with. This latter interpretation requires that we know before hand the basis of these `errors' or projectors. To connect our codes to these interpretations, let's first focus on the case discussed in \ref{pqscc} where the code subspace is completely recovered after the projection. Consider a message $| \psi \ran_m \in \mH_m$ which we choose to encode into the code subspace as
\begin{align}
| \psi \ran_{m} | 0 \ran_{LR} \rightarrow | 0 \ran_m |\wt{\psi} \ran_{LR}
\end{align}
We keep general how much of the state $| \psi \ran$ can be decoded from $L$ or $R$. Since our protocols allow for the information initially in $L$ to be decoded from $R$, the teleportation should be thought as sending part of the message initially encoded in $L$ to $R$. Then, we can append to our physical system an ancilla subsystem $e$ which keeps track of the left measurement:
\begin{align}
 |\wt{\psi} \ran_{LR} | 0 \ran_{e} \rightarrow \sum_{k} P_L^k |\wt{\psi} \ran_{LR} | k \ran_{e}
\end{align}
By measuring $e$ we can determine which projection operator was acted on the physical system and then, assuming the $P_L^k$'s satisfy the conditions of the previous subsections, we can proceed to decode the information of the code subspace.

We can make this look like active quantum error correction by throwing out the information about the ancilla subsystem $e$. Considering a more general state $\wt{\rho} \in \mH_{code} \subset \mH_L \otimes \mH_R$, the evolution of the system is obtained by tracing out $e$ to get 
\begin{align}
\wt{\rho} \rightarrow  \mP \circ\wt{\rho} \equiv \sum_k P_L^k \ \wt{\rho} \ P_L^k 
\end{align}
This evolution is implemented by a quantum channel or a POVM with elements, or `Kraus' operators, $P_L^k$. Assuming that each individual projection can be corrected in the sense of \ref{pqscc}, we can write 
\begin{align}
\mP \circ\wt{\rho} = \sum_k   P^k_L \otimes \left[  U_R W^{P_k}_{\wb{R}_b}  \left( \rho_{\wb{R}_b^1 R_b} \otimes \chi_{\wb{R}_b^2}  \right)  W^{P_k \dagger}_{\wb{R}_b} U_R^\dagger  \right]
\end{align}
Since the different projection operators are orthogonal, we can define a recovery channel with elements
\begin{align}
R_m = P_L^m \otimes W^{P_m \dagger}_{\wb{R}_b} U_R^\dagger
\end{align}
where the projector on the $L$ tensor factor is used to diagnose the error, and the other decodes the message. This clearly decodes the information successfully to give
\begin{align}
\mR \circ \mP \circ \wt{\rho} = \rho_{\wb{R}_b^1 R_b} \otimes \chi_{\wb{R}_b^2}  
\end{align}
where $ \rho_{\wb{R}_b^1 R_b} $ has the same matrix elements as $\wt{\rho}$.

A similar diagnostic procedure can be implemented for the case presented \ref{oap} when  only a subalgebra acting on the code subspace is preserved. The channel $\mO$ acting within the code subspace is preserved or recovered if we can find a corresponding channel $\mO^\mR$ such that
\begin{align}
\mP \circ \mO \circ \wt{\rho} = \mO^{\mR} \circ \mP \circ \wt{\rho}
\end{align}
It's not hard to see that $\mO^\mR$  composed of
\begin{align}
\mO^\mR_m = P_L^m \otimes O^{P_m}_R
\end{align}
would ensure this, where the operator $O^{P_m}_R$ is that constructed in the proof of theorem \ref{them341}.

\section{An Apologia for State Dependence}

We discuss in this section the relation of this framework to previous proposals for the black hole interior \cite{Papadodimas:2012aq, Papadodimas:2013jku, Papadodimas:2015jra, Verlinde:2012cy, Verlinde:2013qya, Verlinde:2013uja}, and address the objections of these proposals raised in \cite{Almheiri:2013hfa,Bousso:2013ifa, Marolf:2013dba, Harlow:2014yoa, Marolf:2015dia} in light of this work. We also comment on the relation of our construction to ER=EPR \cite{Maldacena:2013xja} and provide a possible mechanism for transferring information between two black holes connected via a wormhole. 

\subsection{Arguments Against State Dependence}

We review some of the issues raised against state dependence in \cite{Almheiri:2013hfa, Bousso:2013ifa, Marolf:2013dba, Harlow:2014yoa, Marolf:2015dia} and discuss how they are averted in our construction. Some of these points were already presented in \cite{Papadodimas:2013jku}, for example. All issues here will pertain to large pure black holes in AdS that have come into equilibrium with their Hawking radiation.

{\bf \flushleft $\wt{b}^\dagger$ and the Finite Density of States of the CFT / Typicality}
{\flushleft As discussed in \cite{Almheiri:2013hfa, Marolf:2013dba}, there is a conflict between the algerba of the boundary dressed interior creation and annihilation operators, $\wt{b}_w$ and $\wt{b}^\dagger_w$, and the finite density of states of the dual CFT. The conflict is between the following two statements
\begin{align}
\big[H, \wt{b}^\dagger_w \big] = - w \wt{b}^\dagger_w \ \ \mathrm{and} \ \ \left( {1 \over 1 + \wt{b}^\dagger_w \wt{b}_w} \wt{b}_w \right) \wt{b}^\dagger_w = 1
\end{align}
The first relation is the statement that $\wt{b}^\dagger_w$ lowers the energy of the CFT and is therefore a many-to-one map from the subspace of states of energy $E_0$ to that of energy $E_0 - w$. This reduces the number of states by a factor of $e^{- \beta w}$, which is $\mO(1)$ for $w \sim 1/\beta$. This necessitates that $\wt{b}^\dagger_w$ cannot be an invertible map! However, the second statement shows precisely how the standard low energy QFT algebra ensures the existence of an inverse map. }

As discussed in \cite{Papadodimas:2013jku}, this paradox is easily avoided by taking the interior operators to be state dependent. For example, the operator $\wt{b}^{ \{ s \} \dagger}_w$ associated to the microstate $| B_s^\beta \ran$ will not have the interpretation of a simple mode behind the horizon when acted on another microstate $| B_s'^\beta \ran$ where $s \neq s'$, and will most probably raise the energy of the boundary.

It is interesting to note that conflict does not arise for the the brane dressed versions of $\wt{b}^{ \{ s \} \dagger}_w$, since those do not modify the energy of the boundary to leading order in $N$. Nevertheless, those operators as well are state dependent.

The argument from typicality is also averted by state dependence. In short, the typicality argument involves computing the microcanonical average at some large energy $E_0$ of the Kruskal number operator $N_A = a^\dagger_w a_w$ at the horizon in the basis of Schwarzschild mode number eigenbasis
\begin{align}
\lan N_a \ran_{E_0} = \sum_{n_b} \lan n_b | N_a | n_b \ran
\end{align}
The microcanonical average is basis independent allowing us to choose this particular basis. Now, from the Bogoliubov transformation relating $a_w$ and $b_w$ it is clear that the expectation value of $N_a$ is non-zero in any eigenstate of $N_b$, and therefore
\begin{align}
\lan n_b | N_a | n_b \ran \sim \mO(1)
\end{align} 
The fact that $N_a$ is a positive operator ensures there are no cancellations. This result implies that typical states of the microcanonical ensemble have firewalls.

This argument breaks down for state dependent constructions because the operator $N_a$ is composed of interior operators and therefore is not a linear operator in the Hilbert space that one can simply take the average of. While the previous state dependent constructions want to ensure a smooth horizon for typical states \cite{Papadodimas:2012aq, Papadodimas:2013jku, Papadodimas:2015jra}, we take the perspective that there is no general statement that one can make about arbitrary typical states. We do show how an over-complete basis of typical looking states (where all exterior observables have thermalized) do not have singular horizons.

{\bf \flushleft The Frozen Vacuum and Violations of the Born Rule}
{ \flushleft Another objection to state dependent constructions is the inability of those constructions to find anything else other than the vacuum at the horizon \cite{Bousso:2013ifa}. This criticism does not apply to our construction since the nature of the horizon follows from that of the eternal wormhole prior to the projection, as in section \ref{sec2projlongtfd}.}

Also, the requirement that all typical states have smooth horizons has been shown to lead to violations of the Born rule \cite{Harlow:2014yoa, Marolf:2015dia}. In particular, it is shown how to construct two states, one without a firewall and one with, which are almost parallel in the Hilbert space. This again does not apply in our case since it is not a statement about typical states in general. Consider for example a smooth horizon state, say $|B_s^\beta \ran$, and a unitary $U_s$ which inserts a shockwave just behind the horizon that is dressed to the brane and therefore commutes with the Hamiltonian. We want to interpret the state $U_s | B_s^\beta \ran$ as a black hole with a firewall. Using the techniques of SYK and assuming that the $U_s$ is invariant under the diagonal spin group discussed in section \ref{sec2projlongtfd}, this overlap reduces to the one point function of a unitary $V$ which inserts a shockwave in the TFD:
\begin{align}
\lan B_s^\beta | U_s | B_s^\beta \ran = \lan \beta | V | \beta \ran
\end{align}
This is a one point function in the TFD state and is small if not zero.

\subsection{Relation to State-Dependent Constructions of the Interior}

We first give a quick review of state-dependent constructions of the interior following the formalism of \cite{Papadodimas:2012aq, Papadodimas:2013jku, Papadodimas:2015jra} for definiteness. We will also comment on \cite{Verlinde:2012cy, Verlinde:2013qya, Verlinde:2013uja} which features aspects of QEC.

This proposal is concerned with reconstructing the interiors of large black holes in AdS that have come into equilibrium with their own Hawking radiation. The idea is to begin with a typical state $|\Psi_0 \ran$ drawn from some microcanoncal ensemble at some high energy above the Hawking-Page transition \cite{Hawking:1982dh} of width that doesn't scale with $1/G_N$. Then one considers the algebra of simple operators $\mO_w \in {\cal A}$, written here in fourier modes, dual to a set of low energy operators acting on the exterior of the black hole. ${\cal A}$ is not a closed algebra since it does not include operators composed of products of $1/G_N$ simple operators or larger. This is then used to define a `code subspace' spanned by elements $\mH_{code} = \mathrm{span}\{ {\cal A} | \Psi_0 \ran \}$. Such typical states $| \Psi_0 \ran$ are also called `equilibrium' states  in that correlation functions of operators in ${\cal A}$ are given by their thermal expectation values, as expected from ETH \cite{PhysRevA.43.2046, PhysRevE.50.888, Srednicki:1995pt}. It is then argued that one expects the representation of ${\cal A}$ to be reducible in $\mH_{code}$ allowing for the existence of a nontrivial commutant ${\cal A}'$ of ${\cal A}$. Using the theory of Tomita-Takesaki (see \cite{Witten:2018zxz} for a review), the interior operators $\wt{\mO}_w$ are identified as some subalgebra of ${\cal A}'$ which satisfies the following conditions
\begin{align}
\wt{\mO}_w | \Psi_0 \ran &= e^{-{\beta H\over 2}}\mO^\dagger_w e^{{\beta H\over 2}} | \Psi_0 \ran \\
\wt{\mO}_w \mO_{w_1} ... \mO_{w_n} | \Psi_0 \ran &=\mO_{w_1} ... \mO_{w_n} \wt{\mO}_w | \Psi_0 \ran  \\
[ H, \wt{\mO}_w]\mO_{w_1} ... \mO_{w_n}| \Psi_0 \ran &= w \wt{\mO}_w  \mO_{w_1} ... \mO_{w_n} | \Psi_0 \ran  \\
\end{align}
The operators $\wt{\mO}_w$ are `mirrored' versions of the exterior operators $\mO_w$ defined by these conditions. This construction is motivated by the analogy to the TFD double state,  which due to the entanglement between the left and the right sides we have 
\begin{align}
\mO_L | \beta \ran = e^{-{\beta H\over 2}}\mO^\dagger_R e^{{\beta H\over 2}} | \beta \ran 
\end{align}
for any $\mO_L$ and a corresponding $\mO_R$. From these definitions one finds that  correlation functions involving small numbers of operators from ${\cal A} \cup {\cal A}'$ are given by those in the thermal state, a signal taken to say that the region near the horizon is identical to that in eternal black hole. We therefore see that the construction produces an algebra of interior looking operators whenever observables composed of the simple exterior algebra have all thermalized.

The idea of the interior operators being related to the left operators is in the same spirit as the proposal of this paper. Indeed, the interior operators constructed via QEC satisfy a similar set of constraints as the mirror conditions above
\begin{align}
 {}_L \lan P | \mO_L | \beta \ran_{LR}  &= \wt{\mO}_R \  {}_L \lan P | \beta \ran_{LR} \\
  {}_L \lan P | \mO_L \mO_R^1... \mO_R^n | \beta \ran_{LR}  &= \mO_R^1... \mO_R^n \wt{\mO}_R \  {}_L \lan P | \beta \ran_{LR} 
\end{align}
where $\wt{\mO}_R = \alpha^T \mO \left( \alpha^{-1} \right)^T $ as explained in the previous section. We also have that
\begin{align}
  0 = {}_L \lan P |\big[ O_L, O_R  \big] \mO_R^1... \mO_R^n | \beta \ran_{LR} =    \big[ \wt{\mO}_R, O_R  \big] \mO_R^1... \mO_R^n \ {}_L \lan P | \beta \ran_{LR}
\end{align}
which is just the statement that operators which commute in the unprojected code subspace continue to commute after the projection (assuming both satisfy the recoverability condition of section \ref{oap}). The commutator with the Hamiltonian condition also follows, but it depends on whether the interior operator is dressed to the brane or boundary, where it  will respectively either commute or not.

There are crucial differences though. An obvious one is that the mirroring procedure does not preserve the Hermiticity property of the operators; Hermitian conjugate pairs do not mirror into Hermitian conjugate pairs. In our discussion, this was guaranteed by QEC and proven in section \ref{oap}. It's not clear how much of a problem this is (if at all), but one might worry that since positive operators do not mirror to positive operators in the interior, observables such as the number operator might produce unphysical results in the interior.

The mirroring procedure is reliant on considering an equilibrium state for which all low energy external observables have thermalized. This was not necessary for our construction; we found that we can determine the dictionary both for atypical states of section \ref{sec2projtfd} by projecting on the TFD and for typical states obtained by acting with a series of OTO shockwaves prior to the projection.

Another issue with the mirror construction is that the nature of the interior is determined by the construction rather than by the considered equilibrium state. This was discussed in the previous subsection with regards to the frozen vacuum objection. In our construction the nature of the horizon is predetermined, in part, by the state of the two sided wormhole prior to the projection. We could for instance consider a state which contains a shockwave which skims the horizon from the left hand side and then act with the left projection, just like those in figure \ref{sec2projlongtfd}. In these long wormholes, the right external operators are not sensitive to any of the left shockwaves and, as argued above, will look completely thermalized making such a state indistinguishable from an equilibrium state. Therefore one can carry out the mirroring procedure in this case. However, the actual boundary dual of interior operators will be sensitive to this shockwave while the mirror construction would entirely miss it.

Finally we comment on the use of QEC in \cite{Verlinde:2012cy, Verlinde:2013qya, Verlinde:2013uja} and how it connects to the proposal of this paper. They consider a young black hole not yet maximally (or thermally) entangled with its Hawking radiation and track its state as it emits a single quantum of radiation:
\begin{align}
|\Psi \ran_B | 0 \ran_R \rightarrow \sum_{i} E_i | \Psi \ran_B | i \ran_R
\end{align}
where the state of $B$ belongs to a direct sum of Hilbert spaces of black holes of different masses, and $R$ is the external radiation Hilbert space initialized in the vacuum state $| 0 \ran_R$. This evolution is a unitary transformation acting on the $BR$ system, and therefore the operators $E_i$ must satisfy $\sum_i E_i^\dagger E_i = 1$. Upon tracing out $R$, this evolution looks like the action of an error channel
\begin{align}
\mE\left( | \Psi \ran_B \lan \Psi | \right) = \sum_i E_i | \Psi \ran_B \lan \Psi | E_i^\dagger
\end{align}
Just as in the mirror construction, the goal here is to be able to find the subsystem of $B$ that the radiation state is entangled with and identify it with interior partner Hawking mode. The key result of their work is that if one assumes that this error channel is correctable, i.e. the existence of recovery channel such that
\begin{align}
\mR \circ \mE\left( | \Psi \ran_B \lan \Psi | \right) \propto  | \Psi \ran_B \lan \Psi | 
\end{align}
then one can algorithmically find a subsystem of $B$ which behaves in the appropriate way to mimic the interior Hawking partner. As in standard QEC, this recovery procedure can be implemented on a subspace of states of $\mH_B$, i.e. a code subspace. The recoverability condition becomes $\lan m | E_i^\dagger E_j | n \ran \ \propto \ \delta_{m n}$ for any states $| m \ran$ and $| n \ran$ in the code subspace. However, this proposal again suffers from the same ambiguity issues raised above.

It should therefore be clear that the usage of QEC in this paper and in \cite{Verlinde:2012cy, Verlinde:2013qya, Verlinde:2013uja} is different, though both involve the standard quantum information framework of QEC. The origin of QEC in this paper is the interpretation of the AdS/CFT dictionary as a QEC code. Take for example the discussion of section \ref{toytensorproj}. The representation of the dictionary as a set of tensors, along with the encoding and decoding procedure of going from the bulk legs to the boundary and back, has been proposed as a toy model for the AdS/CFT dictionary by, for example, \cite{Pastawski:2015qua, Hayden:2016cfa}. The goal of the present work was to study how this dictionary is rewired by the application of the projection operator on a subsystem of the boundary.

Nevertheless, it is our view that the proposal of this paper should be viewed as a realization of the general ideas of state-dependent constructions but with more rules so as to stave off some of their inherent ambiguities.

\subsection{Monogamy of Entanglement and ER=EPR}

Next, we engineer situations to satisfy the preconditions of the monogamy of entanglement argument for firewalls \cite{Almheiri:2012rt, Mathur:2009hf, Braunstein:2009my} and see how it affects the nature of the horizon. We will do this by either explicitly considering an entangled state of a set of black hole microstates and some external system or by picking a certain microstate and allowing it to evaporate.

Consider first the set of $2^{N/2}$ black hole microstates $| B_s^\beta \ran_R$, labeled by $s$, of an SYK system $R$ all of which have smooth horizons. This is an overcomplete basis of black hole microstates of effective inverse temperature $\beta$. As discussed in section \ref{section3}, the dictionary between the bulk and boundary is understood for both the exterior modes, $b$, and interior modes $a$, where the dictionary of the latter is state dependent. 

Next, we want to consider entangling $R$ with an external system $E$, which could be another SYK system, in a state $| \Psi \ran_{RE}$ such that the reduced density matrix of $R$ is thermal. This is supposed to mimic an evaporating black hole that has reached the Page time \cite{Page:1993df, Page:1993up} and is thermally entangled with its Hawking radiation. Up to a product unitary $U_R \otimes U_E$ on the two systems, a general such state is
\begin{align}
| \Psi \ran_{RE} = \sum_{s} | B_s^\beta \ran_R | Q_s \ran_E
\end{align}
where ${}_E \lan Q_s | Q_{s'} \ran_E = \delta_{s s'}$. We can check that the reduced density matrix of $R$ is thermal by explicit computation
\begin{align}
\rho_R &= \sum_s | B_s^\beta \ran \lan B_s^\beta | \\
&= e^{-{\beta \over 2} H} \sum_s  | B_s \ran \lan B_s | e^{-{\beta \over 2} H} \\
&= e^{-{\beta} H}
\end{align}
as required. The von Neumann entropy of $\rho_R$ expressed in bulk quantities is
\begin{align}
S(\rho_R) = {A \over 4 G_N} + S_{bulk}(\rho_b)
\end{align}
where $\rho_b$ is the density matrix of the bulk quantum fields $b$. Before we justify this result, we point out that it would satisfy the preconditions of the firewall argument, namely that both the black hole horizon and the external modes $b$ are entangled with the external system $E$. By monogamy of entanglement, this would preclude $b$ from being entangled with the interior modes $a$.

Saying that we now have a firewall is too quick. The reason being that we can take the external system to be another SYK and write its states as
\begin{align}
| Q_s \ran_E = V_E | B_s \ran_E
\end{align}
for some unitary $V_E$, since the states that appear on both sides of this equation are an orthogonal set. Therefore the entangled state between $R$ and $E$ is simply
\begin{align}
| \Psi \ran_{RE}  &= V_E \sum_s | B_s^\beta \ran_R | B_s \ran_E \\
&= V_E | \beta \ran_{RE}
\end{align}
which is just a unitary transformation acting on one boundary of the standard TFD. We see that we have the reverse of the puzzle described in \ref{sec3puzzle}; the modes $a$ were initially encoded on $R$ but have somehow transferred to $E$. The dictionary has been rewired by the entanglement so that $a$ is now reconstructable in $E$. Modulo the unitary $V_E$, the bulk system $b$ continues to be purified by $a$, and their entanglement contributes to the von Neumann entropy of $\rho_R$ in the form of the FLM piece $S_{bulk}(\rho_b)$. Whether there is a firewall or not is determined by the unitary $V_E$. This demonstrates how the fluidity of the dictionary in response to the entanglement realizes the ideas of ER = EPR \cite{Maldacena:2013xja}. We got this by basically going through the SYK projected microstate construction but backwards.

This fluidity can be used to transfer information from $R$ to $E$ by means of entanglement. The basic idea is that prior to entangling $R$ with $E$, we first encode some information in the interior of the pure black hole microstates of $R$ in the modes $a$ via a state dependent unitary
\begin{align}
| B_s^\beta \ran_R \rightarrow U_s^R | B_s^\beta \ran_R
\end{align}
This unitary produces the same density matrix for the bulk fields $a$ for all $s$. Note that this is not a single unitary acted on all the different $| B_s^\beta \ran_R$ but a different one for each state. Entangling these states with the external SYK, but with $V_E = \mI_E$, it's not hard to see that we will get
\begin{align}
| \Psi \ran_{RE} = U^E | \beta \ran_{RE}
\end{align}
where $U^E$ is a truly unitary operator and acts within the code subspace of the eternal black hole on the bulk subsystem $a$. We see that the shift in the dictionary allows us to decode the new state of $a$ from the system $E$ only.

Now, it is a reasonable objection to say that we have not really transferred information from $R$ to $E$, since the encoded information in $R$ was not encoded by a single state independent unitary. Nevertheless, we will now provide a more convincing demonstration of the connection between the transfer of information and the fluidity of the dictionary. We will do this in a series of steps below, but will leave the complete quantitative analysis for future work.

{\bf \flushleft Throwing Information into the Black Hole}
{ \flushleft Consider starting with the TFD state of two SYK systems $L$ and $R$. We can inject some information via a unitary on $R$ at some early time, which proceeds to fall into the black hole.}
\begin{align}
U_R(t_I) | \beta \ran_{LR}
\end{align}
This unitary increases the energy of the right system slightly and takes it out of thermal equilibrium, without changing its von Neumann entropy. After the state thermalizes, it will reach a state where its coarse grained thermal entropy is larger than its von Neumann entropy. We call this difference $\delta S$. 

{\bf \flushleft Evaporation (1/2): Tracking the Trajectory of the Boundary Particle}

{\flushleft Consider then coupling the $R$ SYK to an external auxilliary system $X$ assumed to be at a lower temperature than $1/\beta$ so that energy flows from $R$ into $X$. There are two effects to turning on this coupling which occur in the following sequence. The first is an initial increase of energy of both systems $R$ and $X$, and then a transfer of  energy from $R$ into $X$. }

The initial increase of the energy is explained in appendix \ref{coupling}, and has to do with the fact that, at early times, the leading order effect on the energy comes from the second order contribution in the coupling.  Following this initial spike, the energy starts to leak from system $R$ into system $X$. A good way to model the energy transfer out of the $R$ SYK system is by setting absorbing boundary conditions on the bulk stress tensor along the right boundary \cite{Engelsoy:2016xyb}. In the Schwarzian limit of SYK, the change of energy and the flux of energy at infinity of a massless bulk scalar field theory are related via via
\begin{align}
{d M \over du} = t'^2 T_{t z}
\end{align}
where $t, z$ are bulk Poincare coordinates and $u$ is the boundary time. The energy of $R$ is determined by the boundary trajectory $t(u)$ as
\begin{align}
M = -{\phi_r \over 8 \pi G_N}\{ t, u \}
\end{align}
where $\phi_r$ is the `renormalized' value of the dilaton,  or the coefficient of the growing factor in the bulk dilaton profile as the boundary is approached. 

As discussed in \cite{Engelsoy:2016xyb}, the bulk stress energy due to the Hawking radiation is generated from the conformal anomaly. The relation between the Poincare stress tensor and that due to the black hole is
\begin{align}
T_{y^\pm y^\pm} = (\partial_{y^\pm} x^\pm)^2 T_{x^\pm x^\pm} + {c \over 12} \{ x^\pm, y^\pm  \}
\end{align}
where $x^{\pm} = t \pm z$. The conformal transformation is  $x^{\pm} = x^{\pm}(y^{\pm})$ where $(y^+ + y^-)/2 = u$ is the boundary proper time. $c$ is the central charge of the bulk quantum field theory. The absorbing boundary conditions, in the two coordinate systems, are
\begin{align}
&T_{y^+ y^+} = {c \over 12} \{ x^+, y^+  \}, \ \ T_{y^- y^-} = 0 \\
&T_{x^+ x^+} =0 , \ \ T_{x^- x^-} =  -  (\partial_{y^-} x^-)^{-2 }{c \over 12} \{ x^-, y^-  \}
\end{align}
And therefore, the energy flux on the boundary is
\begin{align}
T_{t z} = {c \over 48} (t')^{-2} \{ t, u \}
\end{align}
This is a negative energy flux falling into the bulk. The energy then satisfies
\begin{align}
{d \over du} \{ t, u \} = - {\pi c G_N  \over 6 \phi_r } \{ t, u \}
\end{align}
This can be solved \cite{Engelsoy:2016xyb} for $t(u)$ to show that the boundary particle receives a (continuous) series of kicks away from the center of the bulk. 

The final precise trajectory of the boundary particle resulting from these two effects depends sensitively on the details of the coupling to the external system. Nevertheless, it is plausible to anticipate that the total effect is to push the boundary particle outwards towards the global AdS$_2$ boundary such that it hits the boundary at an earlier time than the unperturbed situation. We know for sure that it cannot extend beyond this point as that would allow the left SYK  to transmit signals to the right. This would be ensured by the bulk ANEC. It would be interesting to understand the principle on the boundary dual to this\footnote{I thank D. Stanford for discussions on this point.}.

{\bf \flushleft Evaporation (2/2): Tracking the Energy and Entanglement Entropy of $R$}

{\flushleft The energy as a function of $u$ solves to an exponentially decreasing function of time}
\begin{align}
M(u) = M(u_0) e^{- {k }  (u - u_0)}
\end{align}
where $k = {\pi c G_N  \over 6 \phi_r}$ and $u_0$ is the time the absorbing boundary conditions are turned on. Taking $k$ to be small, we can assume the evaporation to be quasi-adiabatic and continue to use the thermodynamic relations between energy, entropy, and temperature. From the energy temperature relation
\begin{align}
M = 2 \pi^2 {\phi_r \over 8 \pi G_N } T^2
\end{align}
the temperature as function of time is found to be
\begin{align}
T(u) = T(u_0) e^{- {k \over 2 }  (u - u_0)}
\end{align}
The thermal entropy as a function of time is 
\begin{align}
S_{th}(u) &= S_0 + 4 \pi^2 {\phi_r \over 8 \pi G_N } T(u) \\
&= S_0 + (S_{th}(u_0) - S_0) e^{-{k \over 2} (u - u_0)}
\end{align}
This thermal entropy is a decreasing function of time and can be thought of as the maximum value of entanglement entropy given the energy $M(u)$. 
\begin{figure}[t]
\begin{center}
\includegraphics[height=5cm]{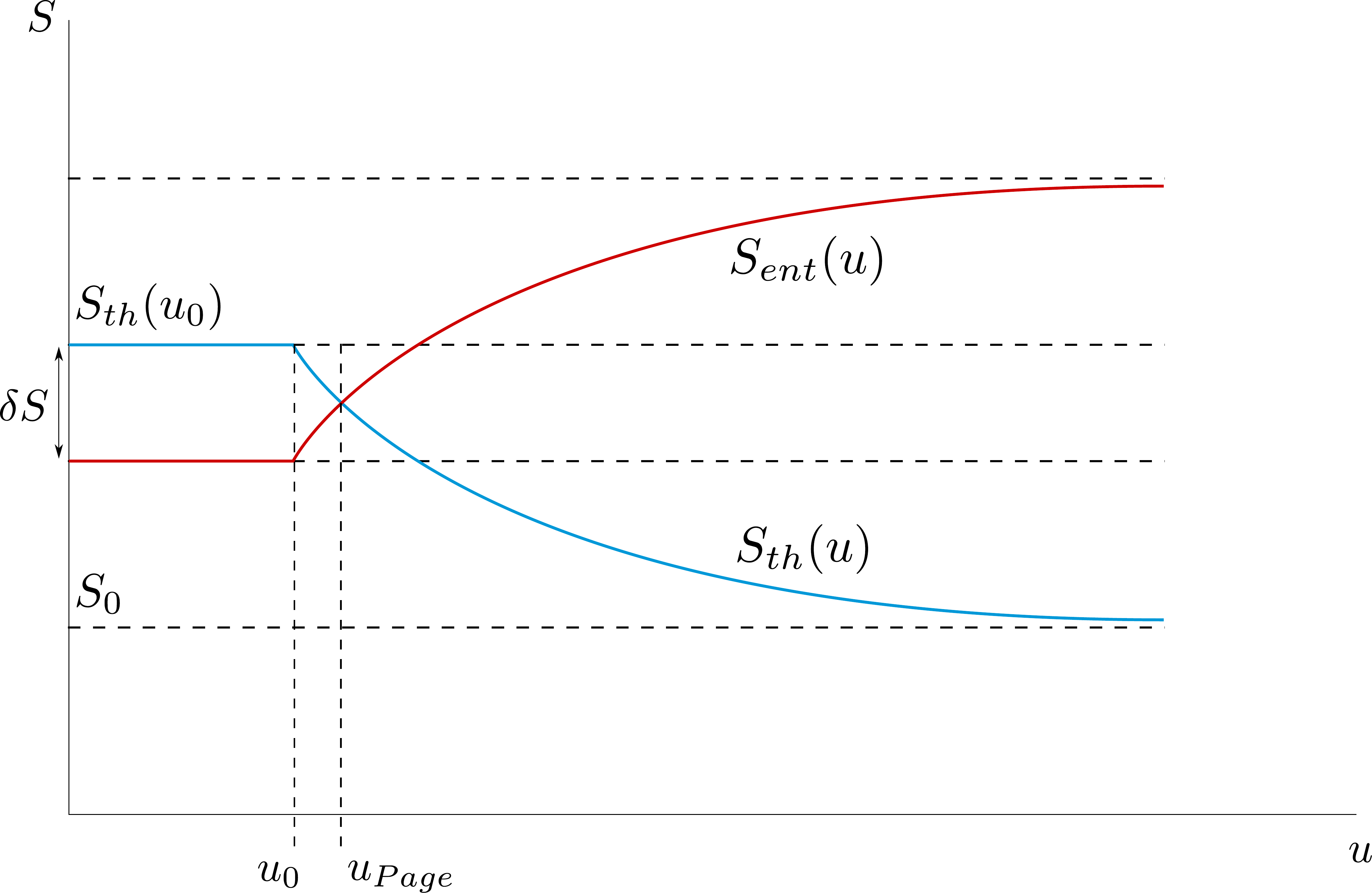}
\caption{Behavior of various entropies as a function of boundary time $u$. The blue curve represents the thermal evolution of the thermal entropy of $R$ computed from the total mass of the black hole, and represents the maximum possible entanglement entropy of $R$. The red curve is the evolution of the entanglement entropy as $R$ leaks energy into $X$, assuming maximal such transfer.}\label{wormholepagecurve}
\end{center}
\end{figure}

The insertion of the message at early times  increases the energy, and therefore the von Neumann entropy of $R$ differs from its thermal entropy by $\delta S$, 
\begin{align}
S_{ent}(u_0) = S_{th}(u_0) - \delta S
\end{align}
As the black hole evaporates, we can model the increase of the entanglement entropy by the decrease of the thermal entropy, which follows from the usual state of Hawking radiation. Again, assuming quasi-adiabatic evaporation we can write
\begin{align}
\Delta S_{ent}(u) = - \Delta S_{th} = - \int_{u_0}^u {dE(u) \over T(u)}  = (S_{th} (u_0) - S_0) \left( 1 - e^{- {k \over 2} (u - u_0)}\right)
\end{align}
where now we have
\begin{align}
S_{ent}(u) &= S_{ent}(u_0) + \Delta S_{ent}(u) \\
&= S_{th}(u_0) - \delta S + (S_{th} (u_0) - S_0) \left( 1 - e^{- {k \over 2} (u - u_0)}\right)
\end{align}

Just as in the standard Hawking evaporation in any dimension, the analogous relation obtained from the usual Hawking process is only trustworthy until around the Page time. This is the time when the thermal entropy of the system is equal to its entanglement entropy,
\begin{align}
S_{th}(u_{Page}) = S_{ent}(u_{Page})
\end{align}
This time is
\begin{align}
u_{Page} &= u_0  - {2 \over k} \ln \left[ 1 - {\delta S \over 2 (S_{th}(u_0) - S_0)}\right] \\
&\approx u_0 + {12 \over c} {\delta S \over T(u_0)}
\end{align}
Which is a short time for $\delta S \sim \mO(1)$. Starting at this time, the thermal entropy of the black hole will be equal to its entanglement entropy. This will be given by the area of the new horizon of the smaller black hole, which therefore becomes the RT surface for the entire system $R$. We expect this to follow since the density matrix of $R$ approaches the thermal state. See figure \ref{wormholepagecurve}.

{\bf \flushleft Deposit the Extracted Energy of $R$ from $X$ into $L$}

{\flushleft After transferring energy from system $R$ to $X$, the state of $LR$ is no longer pure. We gain extra information about the nature of this state by evolving both systems using the original time independent Hamiltonians to the far future and far past. We expect that the right boundary particle will, again, hit the global boundary prematurely. This will result in a new horizon for $R$ that must be its new RT surface since $\rho_R$ is (almost) thermal. This is shown in the third diagram of figure \ref{wormholetransfer}.  The bulk dual of the $LR$ system will then contain two RT surfaces, one for each boundary. These are the surfaces (points)  $A_L$ and $\wt{A}_R$, which do not coincide. }
\begin{figure}[t]
\begin{center}
\includegraphics[height=3.9cm]{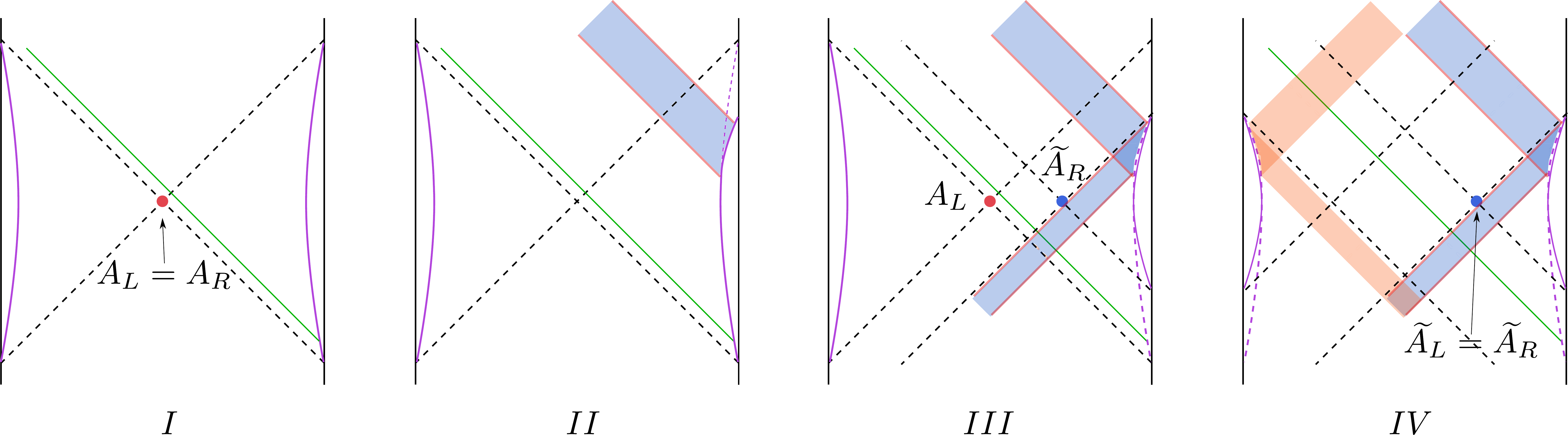}
\caption{The left most diagram ($I$) represents  dropping a message (green) into the TFD of $LR$. The second ($II$) represents the extraction of energy from $R$, represented here by the some initial positive energy (red) followed by negative energy (blue) and some final positive energy due to switching off the interaction (red). In the third diagram ($III$), the right system is evolved using a time independent Hamiltonian. Here the space time has two RT surfaces (points), $A_L$ and $\wt{A}_R$ for the $L$ and $R$ respectively. In the last diagram ($IV$) the RT surfaces coincide again but at the new location of $\wt{A}_R$. The initial message is in the entanglement wedge of the $L$ (but outside its causal wedge).}\label{wormholetransfer}
\end{center}
\end{figure}

Consider now depositing the energy extracted from $R$ into $L$ by means of a unitary acting on $LX$. We imagine that this process can be done in a quasi-adiabatic way on $L$ so as to not modify the bulk picture drastically. This process will not alter the density matrix of $R$, and therefore $\wt{A}_R$ will continue to be its RT surface. After this process is complete, the state of $LR$ will be pure and the RT surfaces of the two boundaries will coincide on $\wt{A}_R$.

Something interesting has just happened. The message sent into $R$ at early times is now contained within the entanglement wedge of $L$, and thereby reconstructable from $L$. We see that the fluidity of the dictionary under entanglement transfer has been rewired the dictionary precisely such that information initially in $R$ is now contained in $L$.  We note the parallel here between this information transfer and the Hayden-Preskill criterion for the decoding a message from the Hawking radiation \cite{Hayden:2007cs}. Here we view $R$ as the black hole and $L$ as the Hawking radiation, and the TFD as the system at the Page time. Then, after throwing in a new message, we have to wait for some extra time for qubits to transfer from $R$ to $L$ until $R$ becomes maximally entangled with $L$, at which point the message can be decoded from $L$.

A similar observation can be made for an evaporating pure large black hole in AdS$_2$, although there is a subtlety due to its ground state entropy. As reviewed earlier, for this black hole to reach the Page time it must build up its entanglement entropy until it coincides with its thermal entropy $S_{th}(T)  = S_0 + C T$, where $C$ is some constant. However, since the rate of evaporation is controlled by the temperature, $\dot{M} \sim T^2$, it is clear that the black hole can only evaporate away a $C T_\mathrm{initial}$ amount of entropy, and therefore will never become thermally entangled with the auxiliary system\footnote{Since the ground state degeneracy is actually lifted by the SYK interactions, the black hole will actually evolve to the Page time provided we wait long enough. However, the Schwrazian description is expected not to be valid for such long times.}. Naively, one might have thought possible to consider large temperatures such that $C T_\mathrm{initial} > 2 S_0$, though this is not the case since then $S_{th}$ would exceed the total number of states of the SYK system; for large $q$,  $S_0 = {N \over \ln 2}  - N \pi^2/4 q^2 + ...$ \cite{Maldacena:2016hyu, PhysRevB.63.134406} and the total number of states in the SYK model is $2^{N/2}$.

This subtlety can be avoided by  adding  another process which continuously adds pure matter into the evaporating black hole so as to keep its temperature constant. The entanglement via the Hawking process will continue to increase and the combined effect of the added matter and the evaporation will push the boundary particle outwards towards the global AdS$_2$ boundary. Just as in the wormhole example, a message thrown in at early times will escape the new RT surface generated by the build up of entanglement. In order to retrieve the information, one can imagine depositing the extracted energy into another SYK system, $L$, and, up to a unitary on this system alone, the dual spacetime can be made to look like a wormhole with the message located in the entanglement wedge of $L$.

\subsection{Comments on Complexity}

\begin{figure}[t]
\begin{center}
\includegraphics[height=5cm]{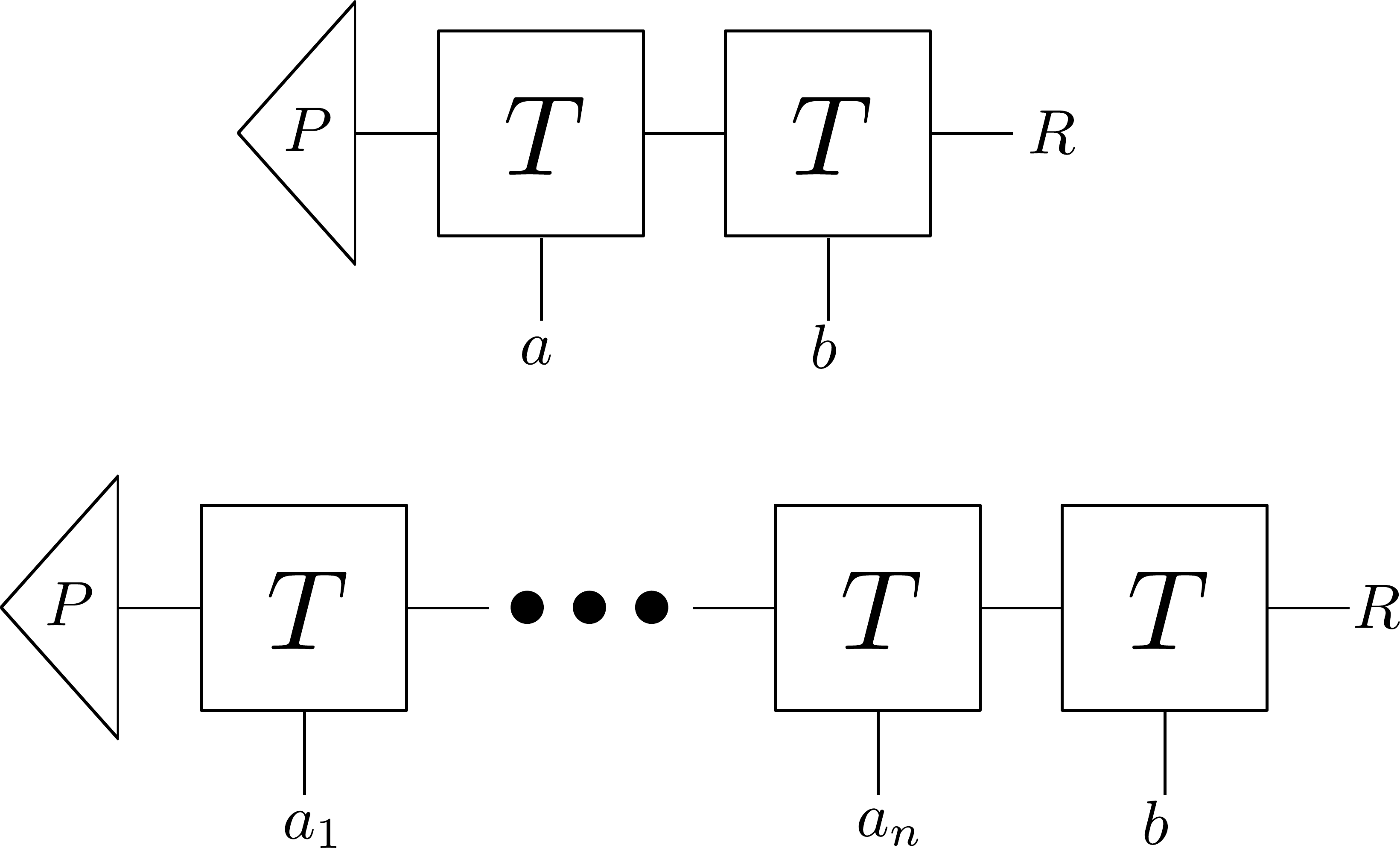}
\caption{The projected tensor network can be viewed as a circuit which prepares the state of the right boundary $R$ assuming some inputs from the bulk $a$ and $b$ and starting from a very simple state $P$ on the left (a product state of spins). The top figure corresponds to projecting on the thermofield double which gives a relatively short tensor network and therefore prepares a low complexity state on $R$. The bottom figure shows a long tensor network which prepares a complex state for large $n$ (the number of tensors).}\label{complexitynetwork}
\end{center}
\end{figure}

We comment here on the connections with the idea of holographic computational complexity \cite{Susskind:2014rva, Stanford:2014jda, Susskind:2014jwa, Susskind:2015toa}. It is interesting to note the difference of complexity between the typical and atypical black hole microstates considered in section \ref{projectedmicrostates} and how that depends on the details of the projector. The projection operator $|B_s \ran \lan B_s |$ projects the left SYK onto a simple product state of spins - a state of low complexity. The resulting state of the right SYK is the Euclidean time evolution of this product state by an amount $\beta/2$
\begin{align}
| B_s^\beta \ran = e^{- {\beta \over 2} H} | B_s \ran
\end{align}
Assuming that $\beta \sim \mO(1)$, we will take this evolved state to be roughly of the same complexity as the product state, that is both are simple\footnote{Euclidean evolution generically takes all states to the vacuum, assuming we evolve for long enough, and therefore tends to a complexity decreasing transformation.}. Moreover, the simplicity of this state can be deduced from the relatively short projected tensor network, which can be viewed as that which prepares the state of $R$ starting with a simple state of $P$.

Projecting a long wormhole supported by left OTO shockwaves waves does not produce such a simple state. As shown in figure \ref{complexitynetwork}, projecting on a long wormhole constructed by $n$ OTO shockwaves results in a  tensor network composed of roughly $n$ tensors. Each tensor is generated by sandwiching the insertion of a local operator, the shockwave, by Hamiltonian evolution of a scrambling time. Taking into account the partial cancellation between the forward and backward time evolution, we can estimate the complexity of each tensor as $N$, the number of spins in SYK, and therefore the total complexity of these states is roughly $n \times N$\footnote{In a previous draft we forgot to take into account the partial cancellation between the forward and backward time evolution and concluded the complexity of each tensor to be $\ln N$. We thank Ying Zhao for pointing this out to us.}\cite{Susskind:2014rva, Stanford:2014jda, Susskind:2014jwa, Susskind:2015toa}.

\begin{figure}[t]
\begin{center}
\includegraphics[height=3.9cm]{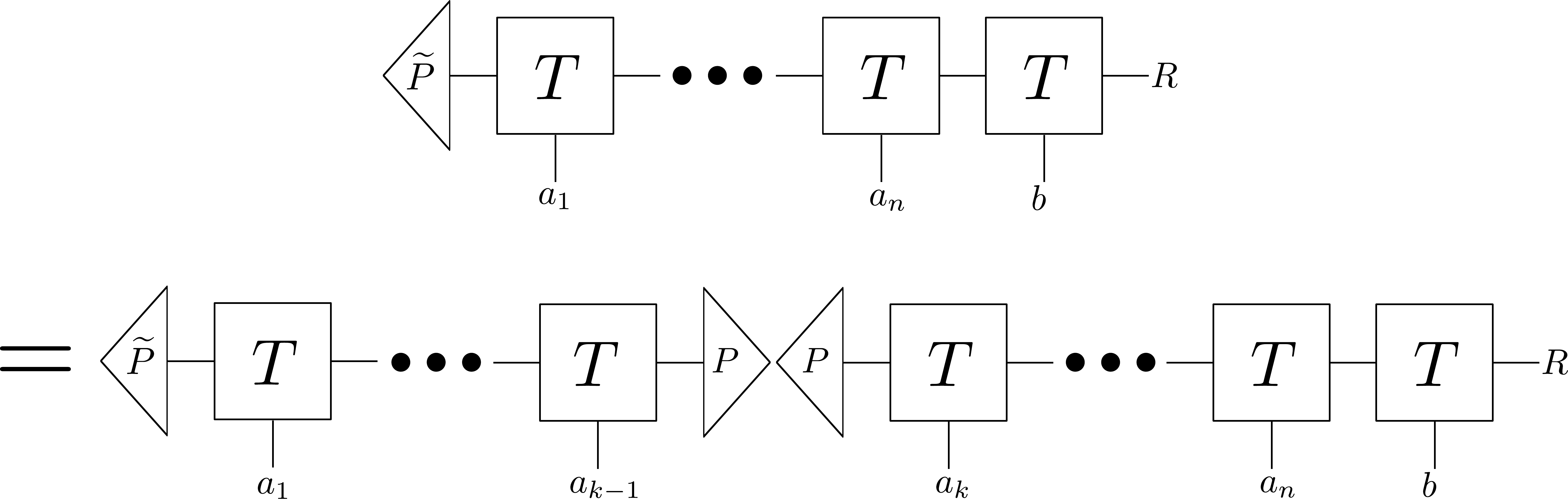}
\caption{Projecting on the left SYK with a complicated but fine tuned projector $\widetilde{P}$ projects a chosen set of bulk internal indices into a simple state $P$. To estimate the complexity of the right SYK $R$ we only need to count the number of tensors between it and the simple projector, which in the presented case is $n-k+1$ tensors. }\label{finetunedcomplexity}
\end{center}
\end{figure}

Counterintuitively, it turns out that projecting on the left SYK with a more complicated state of the spins can result in a simpler state of $R$. The caveat is that this more complicated projector needs to be fine tuned with respect to the bulk tensors. A demonstration of this is to start with a long tensor network where we can fine tune the projector to any given number of bulk tensors and shorten the network to any length that we desire, up to the horizon\footnote{The reason we can't go past the horizon in these tensor network models is that the horizon acts as a bottleneck and thus only the bulk legs  to its left  (internal legs included) map isometrically to the left boundary. Similarly, all bulk legs on the right of the horizon map to the right boundary.}. The key is to find the left projector which projects the bulk internal legs between a given pair of tensors into a simple state. In the situation where the black hole is lengthened by a series of OTO shockwaves, this can be easily achieved by picking a projector which undoes these shockwaves. It would be interesting to understand this in the more general setup where  the lengthening procedure is not so simple.

\section{Conclusion}

The goal of this paper was to understand the dictionary for operators inside the horizon of pure black hole microstates. We considered such microstates in the SYK model which are prepared by starting with the thermofield double, dual to the eternal black hole, and completely projecting out one of the boundaries \cite{Kourkoulou:2017zaj}. The dual of this projection is to insert an end-of-the-world brane near the projected boundary which falls into the black hole. This prepares an overcomplete set of black holes all of which are firewall-free. 

We argued that this preparation process would naively create a firewall at the bifurcation surface. The point was that the entanglement of the bulk fields across the horizon contributes to the entanglement entropy between the two boundaries \cite{Faulkner:2013ana}, and one might worry that breaking the latter would necessarily break the former. This would naively follow  from subregion-subregion duality which says that the density matrix of the bulk fields on the left/right can be recovered from the density matrix of the left/right boundary.

Nevertheless, we showed that the quantum error correction interpretation of the duality avoids this conclusion by giving the AdS/CFT dictionary an interesting kind of fluidity. We showed how the (say left) projection causes a rewiring of the dictionary so as to map bulk operators originally dual to the left boundary to the right. This establishes a dictionary for operators behind the black hole horizon.

This dictionary was found to have the interesting feature that it depends on the projection operator used. This is reminiscient of previous state-dependent proposals of reconstructing the black hole interior \cite{Nomura:2014woa, Nomura:2013gna, Nomura:2012cx, Verlinde:2013uja, Verlinde:2013qya, Verlinde:2012cy, Papadodimas:2015jra, Papadodimas:2013jku, Papadodimas:2012aq}. We comment that a key difference between this work and these proposals is that our construction first considers a bulk state where the nature of the horizon is known and then finds the dictionary, while the previous proposals begin with a boundary equilibrium state and then constructs a subalgebra on the boundary for which the black hole horizon looks smooth.  We show that one can construct an explicit example of a black hole that looks to be completely thermalized (an equilibrium state) from the exterior but which has a `firewall' just behind the horizon which these state-dependent constructions would entirely miss.

We also preformed a preliminary analysis of how to utilize the fluidity of the dictionary to transfer information between two black holes connected by a wormhole. By starting with two SYKs in the TFD state, we showed that extracting energy from one boundary and dumping it in the other causes the RT surface to shift to a new surface of smaller area that is spacelike related to the original RT surface and positioned between it and the boundary. That there should be a new RT surface of smaller area follows because the temperature of the evaporating black hole is decreasing and so must its entanglement entropy, and therefore the original RT surface would suggest a larger entropy than is allowed by thermodynamics. We showed how this implies that a message sent in at early times from the evaporating side ends up within the entanglement wedge of the growing side and thereby becomes reconstructable from the other boundary. This occurs once the evaporating side has reached the Page time, when its entanglement entropy equals its thermal entropy.

We also argued for the analogous effect for the case of starting with a  pure large black hole in AdS and allowing it to evaporate. We can engineer the situation so that the black hole reaches the Page time and becomes thermally mixed with some external system. Once again, a message sent in at early times will be located outside the newly generated RT surface, and therefore will  not be reconstructable on the original system. The principle in play in both of these examples is that a black hole allowed to evaporate via a generic non-fine-tuned process will have its event horizon coincide with its RT surface by the Page time. This presents a new picture of the evolution of the spacetime as a black hole evaporates shown in figure \ref{evapbh}.

\begin{figure}[t]
\begin{center}
\includegraphics[height=3cm]{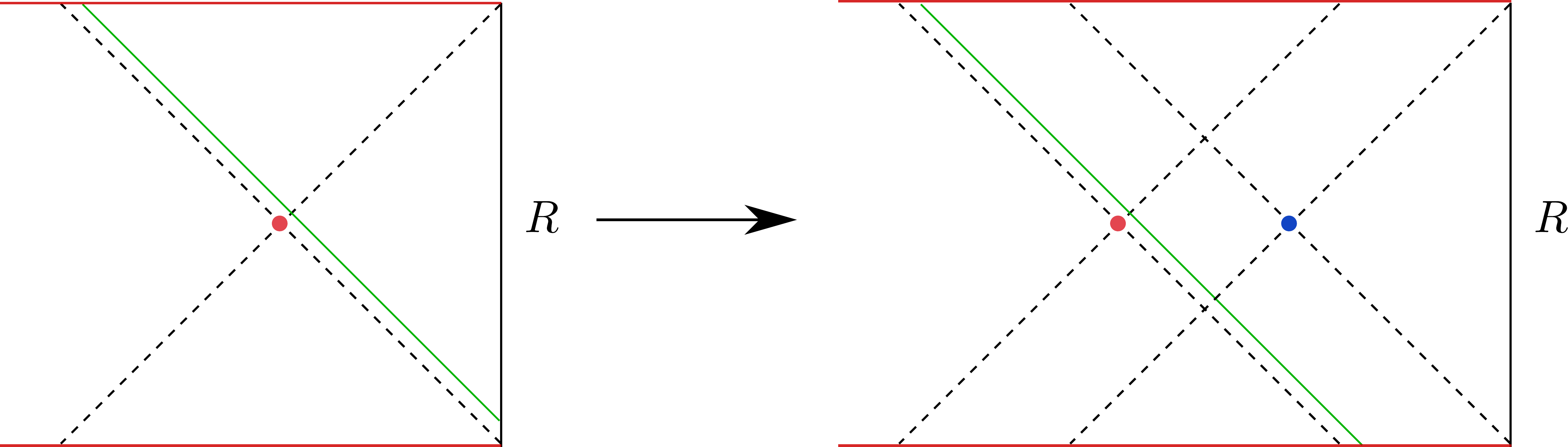}
\caption{The state of an evaporating black hole is dual to a longer and longer wormhole, with ever changing RT surface. The blue dot is the new RT surface for the now smaller black hole.}\label{evapbh}
\end{center}
\end{figure}

We also commented on how the choice of the left projector determines the complexity of the resulting state on the right. By focusing on the tensor network representation of $LR$ system and projecting on $L$, we argued that the resulting tensor network can be viewed as the quantum circuit which prepares the state of $R$ and from which the complexity of the state can be estimated.

\acknowledgments

I would like to thank Juan Maldacena for the many illuminating discussions which led to this work, and Douglas Stanford for being a constant resource along the way. I would also like to thank  William Donnelly, Patrick Hayden, Nima Lashkari, Raghu Mahajan, Geoffrey Penington, David Poulin,  Xiao-liang Qi, Herman Verlinde, and  Ying Zhao for discussions. I are grateful to the KITP Program Quantum Physics of Information (Sep 18 - Dec 15, 2017), where some part of the work was performed. This research was supported in part by the National Science Foundation under Grant No. NSF PHY17-48958.

\appendix

\section{Bulk Particle Gravitational Dressing}
\label{gravdressing}

Let's begin by listing a set of coordinates for AdS$_2$:
\begin{align}
Y^{-1} &= {\cos t \over \sin \sigma} = \cosh r \\
Y^{0} &= {\sin t \over \sin \sigma} = \sinh r \sinh \tau \\
Y^{1} &=  - {\cos \sigma \over \sin \sigma} = \sinh r \cosh \tau
\end{align}
Which are the embedding, global, and Rindler coordinates respectively. The metrics are the following
\begin{align}
ds^2 &= - \left( d Y^{  -1} \right)^2 -  \left( d Y^{ 0 } \right)^2 +  \left( d Y^{  1} \right)^2, \ Y^2 = -1 \\
ds^2 &= {-dt^2 + d \sigma^2 \over \sin^2 \sigma} \\
ds^2 &= dr^2  - \sinh^2 r \ d\tau^2
\end{align}
The trajectory of a massive particle is completely determined by the condition
\begin{align}
Y \cdot Q = 0
\end{align}
For a particle that sits in the center of the bulk the charge is given by
\begin{align}
Q^a_{center} = (0, 0, -m)
\end{align}
Via an SL2 transformation we can push this particle to any massive geodesic. The most general form for the charge of such a particle is
\begin{align}
Q^a = m( \sinh \gamma \sin \theta, \sinh \gamma  \cos \theta, -\cosh \gamma)
\end{align}
where $\gamma$ can be thought of as a rapidity determining the velocity of the particle when it passes the center of the bulk, and $\theta$ controls the shift of the trajectory in bulk global time. The trajectory of the particle in embedding coordinates is
\begin{align}
Y^{-1} &=  \cos \theta \cos T  +\sin\theta \cosh \gamma  \sin T \\
Y^{0} &= -\sin  \theta \cos T + \cos \theta  \cosh \gamma  \sin T  \\
Y^{1} &= -\sinh \gamma \sin T
\end{align}
where $T$ is some time parameter along the trajectory.

The brane of \cite{Kourkoulou:2017zaj} reaches the boundary at bulk time $t = \tau = 0$. In embedding coordinates this is
\begin{align}
Y^{-1} &\rightarrow \infty \\
Y^{0} &= 0  \\
Y^{1} &\rightarrow -\infty
\end{align} 
We then deduce the values of $\theta$ and $\gamma$ to be
\begin{align}
\gamma &\rightarrow \infty \\
\theta &= {\pi \over 2}
\end{align}
and embedding time parameter
\begin{align}
T =  {\pi \over 2}
\end{align}
This is the embedding proper time at which the particle is a maximum $|Y^1|$.

Next we turn to the bulk particle. For a particle to fall into the black hole from the left exterior we have
\begin{align}
0  &<  \gamma < \infty \\
0 &<  \theta < {\pi} 
\end{align}
The first condition ensures that the particle is neither at rest nor falling in at the speed of light. The second ensures the particle falls in from the left exterior by guaranteeing that the largest radial position of the particle (the point where $T = {\pi \over 2}$) occurs within $-{\pi \over 2} < t < {\pi \over 2}$, where $t$ is the bulk global time. At this point $t$ and $\theta$ are related via
\begin{align}
t = {\pi \over 2} - \theta
\end{align}

Now we consider the boundary particle. It's trajectory is fixed by the condition 
\begin{align}
Y \cdot Q = - q
\end{align}
for some $q$. It turns out that $Q$ is proportional to the location of the bifurcation point in embedding coordinates. We have been working in the gauge where the bulk $t = 0$ slice corresponds to the $Y^0 = 0$ slice in embedding coordinates. Therefore, we can ensure the bifurcation point also rests on this slice by picking the charge to be
\begin{align}
Q^a_{R_\partial} = (\sqrt{E}, 0, 0)
\end{align}
where $R_\partial$ is the label for the right boundary particle. Assuming for now that we have the thermofield double, we would require another boundary particle for the left side whose charge must be
\begin{align}
Q^a_{L_\partial} = (-\sqrt{E}, 0, 0)
\end{align}
by the requirement $Q^a_{L_\partial} + Q^a_{R_\partial} = 0$. The energy as measured on the right boundary is simply the square of the charges
\begin{align}
H = -Q^2_{R_\partial}  = E
\end{align}
before considering the brane, we can study how the charges and trajectories of the TFD get modified by the presence of a bulk particle. The bulk particle charge is
\begin{align}
Q^a_{Bp} = m ( \sinh \gamma \sin \theta, \sinh \gamma  \cos \theta, - \cosh \gamma)
\end{align}
And we need to satisfy
\begin{align}
Q^a_{L_\partial} + Q^a_{R_\partial}  + Q^a_{Bp} = 0
\end{align}
There are obviously an infinite number of ways to do this, and they correspond to how the bulk particle is dressed to either boundary. Two interesting cases is when the particle is either entirely dressed to the right:
\begin{align}
Q^a_{R_\partial} &= ( \sqrt{E} - m  \sinh \gamma \sin \theta,  -m \sinh \gamma  \cos \theta,  m \cosh \gamma  ) \\
Q^a_{L_\partial} &= (-\sqrt{E}, 0, 0)
\end{align}
or entirely to the left
\begin{align}
Q^a_{R_\partial} &= (\sqrt{E}, 0, 0) \\
Q^a_{L_\partial} &=  ( -\sqrt{E} - m  \sinh \gamma \sin \theta,  -m \sinh \gamma  \cos \theta,  m \cosh \gamma  ) 
\end{align}

The final thing to show in the context of the TFD is how the bulk matter affects the boundary particle trajectory. I will show that independent of $\theta$, which exterior the bulk particle emerges into, the boundary particle is pushed towards the global boundary and hits it sooner compared to the no bulk particle case. In the case with no bulk particle, the boundary particle trajectory is bounded between the global bulk times 
\begin{align}
-{\pi \over 2} \le t \le {\pi \over 2}
\end{align}
The goal is to show that the modified trajectory is bounded by as
\begin{align}
-{\pi \over 2} < t_- \le  t \le t_+ < {\pi \over 2}
\end{align}
where the $t_-$ and $t_+$ are the new boundary times which the boundary particle approaches. To see this, we have to note that the boundary particle trajectory requires that $Y \cdot Q $ be a constant. The idea is that by taking the boundary limit while keeping this quantity fixed we should find that the global time approaches a certain value. Recall that we can reexpress the embedding coordinates in terms of global coordinates as
\begin{align}
Y^{-1} &= {\cos t \over \sin \sigma} \\
Y^{0} &= {\sin t \over \sin \sigma}\\
Y^{1} &=  - {\cos \sigma \over \sin \sigma} 
\end{align}
All of which diverge at the same rate as $\sigma \rightarrow \pi$. Therefore we find the condition that
\begin{align}
-Q^{-1} \cos t - Q^{0} \sin t + Q^{1} = 0
\end{align}
where 
\begin{align}
Q^a &= ( \sqrt{E} - m  \sinh \gamma \sin \theta,  -m \sinh \gamma  \cos \theta,  m \cosh \gamma  ) 
\end{align}
In the case with no bulk particle we have $Q^a = ( \sqrt{E} ,0,0  ) $ and therefore $t = \pm {\pi \over 2}$ satisfies the constraint. The general solution of the constraint is
\begin{align}
\cos t_{\pm} = {Q^{-1} Q^{1}  \pm |Q^0| \sqrt{H} \over \left( Q^{-1}\right)^2 + \left( Q^{1}\right)^2}
\end{align}
Where $H = -Q^2$. Let's evaluate this for $m \cosh \gamma \ll E$. To first order in $m \cosh \gamma/\sqrt{E}$ we find
\begin{align}
\cos t_{\pm} = {m \over \sqrt{E}} \left[ \cosh \gamma  \pm \sinh \gamma \cos \theta \right] > 0
\end{align}
and therefore $|t_{\pm}| < {\pi \over 2} $.

Finally we consider the case of the brane. To analyze this we need to consider the covering space before the $Z_2$ has been taken \cite{Kourkoulou:2017zaj}. In this space we would have the brane, which we can place at rest in the center, and two boundary particles. The extra bulk particle also needs to be duplicated. The brane and bulk particle charges are
\begin{align}
Q^a_{brane} &= (0 , 0, -\mu) \\
Q^a_{BP_L} &= m ( \sinh \gamma \sin \theta, \sinh \gamma  \cos \theta, - \cosh \gamma)
\\
Q^a_{BP_R} &=  m (- \sinh \gamma \sin \theta, -\sinh \gamma  \cos \theta, - \cosh \gamma)
\end{align}
\begin{figure}[t]
\begin{center}
\includegraphics[height=6cm]{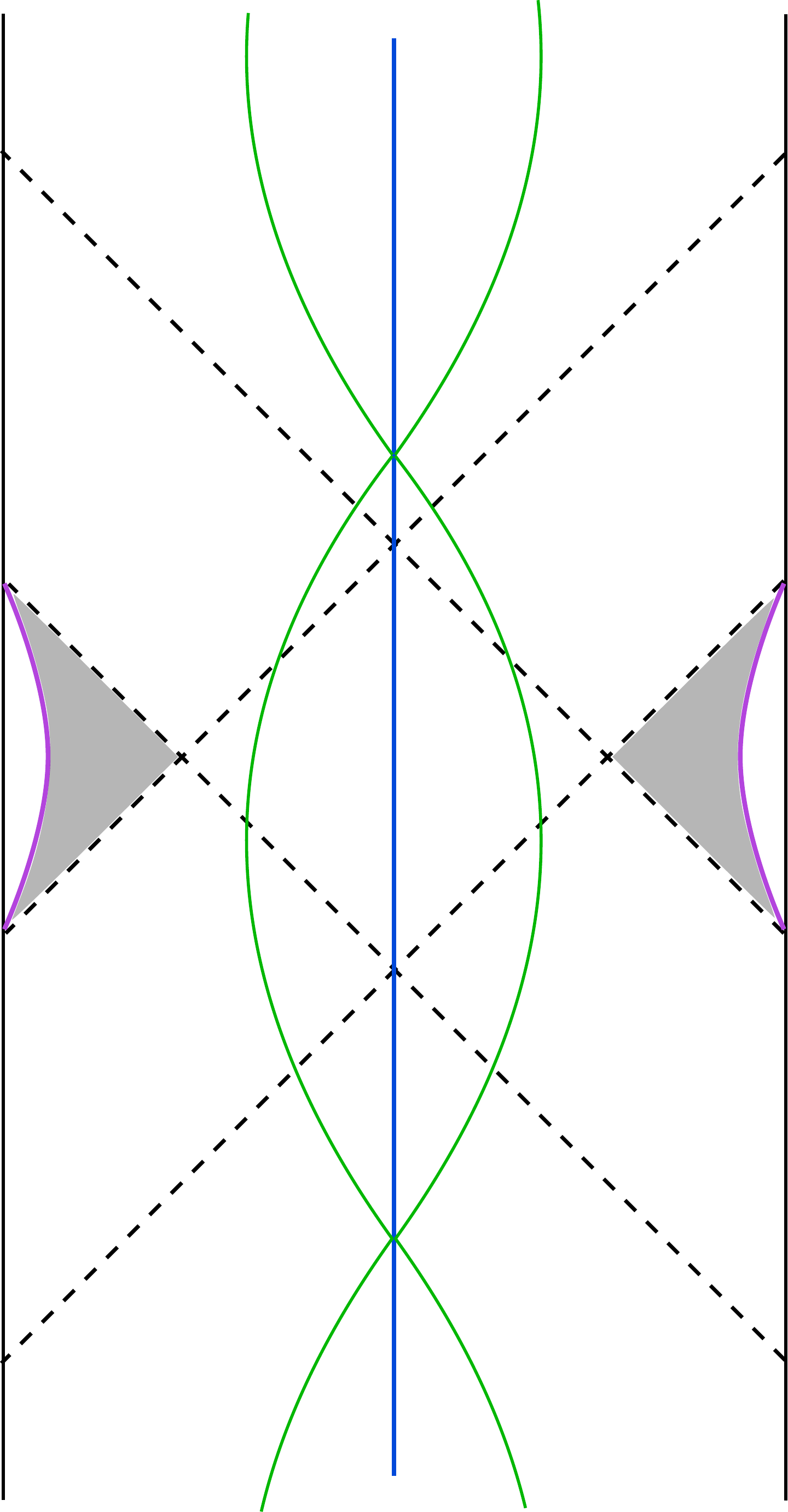}
\caption{The covering space of the EWB geometry with a bulk particle. The covering space needs to be considered to ensure the vanishing of the gauge constraint.}\label{braneparticlez2}
\end{center}
\end{figure}
were the two bulk particle charges are related by $\theta \rightarrow \theta + \pi$. Notice that the sum of the bulk particle charges is
\begin{align}
Q^a_{BP_L}  + Q^a_{BP_R} = (0 , 0, - 2 m \cosh \gamma)
\end{align}
and again, now we have the choice to either dress the bulk particles to the brane or the boundary particles. The analysis of the latter case is identical to what we did previously with the TFD. Dressing them to the brane simply changes the mass of the brane to
\begin{align}
-\mu \rightarrow - \mu  + 2 m \cosh \gamma
\end{align}
Note that the boundary particle trajectories are unchanged since the total charge of the bulk particles plus the brane is equal to that of the previous case with only a brane by itself (with a different mass). Also, the trajectory of the brane is unaltered since it is insensitive to multiplying the charge by an overall factor. Note that brane mass is always decreased independent of $\theta$.

\section{Initial Energy Increase due to an External Coupling}
\label{coupling}

We prove in this appendix that at early times that coupling to another system will generically raise the energy of an initially static system. We imagine coupling a SYK system $R$ to an external system $X$.

Let  $X$ be at arbitrary system in its vacuum, and let's consider modifying the total Hamiltonian with a general coupling 
\begin{align}
H_T =  H_0 + \delta H(t)
\end{align}
where
\begin{align}
H_0 &\equiv H_{SYK}^R + H^X \\
\delta H(t) &\equiv  \lambda A_R(t) B_X(t)
\end{align}
Let's assume that $R$ is in the thermal state. Working in the interaction picture, the evolved state is
\begin{align}
|\Psi(t) \ran_{LRX} = e^{-i H_0 t} {\cal T} e^{- i \int_{t_0}^{t} dt' \delta H(t')} | \beta \ran_{LR}|0 \ran_{X}
\end{align}
We want to compute change in energy of system $R$ immediately after turning on the interaction, which we do so by working to leading order in $\delta t = t - t_0$. We can then Taylor expand the interaction exponent
\begin{align}
|\Psi(t) \ran_{LRX} = e^{-i H_0 t} e^{- i \delta t \delta H(t_0) - i \delta t^2 \partial_{t_0} \delta H(t_0) + ...} | \beta \ran_{LR}|0 \ran_{X}
\end{align}
Now we'll compute the instantaneous energy change for either system. The general computation is
\begin{align}
\delta E_K &= {}_{LRX} \lan \Psi(t)| H^K  |\Psi(t) \ran_{LRX} - {}_X \lan 0 | {}_{LR} \lan \beta | H^K | \beta \ran_{LR} |0 \ran_{X} 
\end{align}
where $K = \{ R, X \}$. The first order in $\delta t$ comes from 
\begin{align}
\delta E^{(1)}_R &= - i \delta t  \  {}_X \lan 0 | {}_{LR} \lan \beta | \ [ H^R_{SYK}, \delta H(t_0) ] \  | \beta \ran_{LR}|0 \ran_{X} \\
&=  - \lambda \delta t  \   \lan \dot{A}(t) \ran  \ \lan B(t) \ran \\
\delta E^{(1)}_X &= - i \delta t  \  {}_X \lan 0 | {}_{LR} \lan \beta | \ [ H^X_{\phi}, \delta H(t_0) ] \  | \beta \ran_{LR}|0 \ran_{X} \\
&=  - \lambda \delta t  \   \lan {A}(t) \ran  \ \lan \dot{B}(t) \ran
\end{align}
where $ \lan A(t) \ran =  \lan \beta | A(t) | \beta \ran$ and $\lan 0 |  B(t) | 0 \ran$. Note that both states of $X$ and $R$ are time translation invariant, and therefore the time derivatives of one point functions must vanish. The same conclusion holds assuming that $R$ is not precisely the thermal state but has thermalized. We conclude that to first order $\delta E = 0$, or atleast to very good approximation. Not that it would have been problematic if this wasn't true since we have the freedom to tune the sign of $\lambda$ so as to reduce the energy of the scalar field theory below that of the vacuum. 

We turn next to the second order contribution in $\delta t$. Expanding, we find
\begin{align}
\delta E^{(2)}_R = {\lambda^2 \over 2} i \lan  [\dot{A}(t_0), A(t_0)] \ran \lan B^2(t_0) \ran \\
\delta E^{(2)}_X = {\lambda^2 \over 2} i  \lan A^2(t_0) \ran \lan [\dot{B}(t_0), B(t_0)] \ran 
\end{align}
Note that we do not have a choice in the overall sign of these contributions to the total energy. We first give a qualitative argument for why these contributions have to be positive, and then prove it rigorously. For either system, these expressions are what one would obtain when turning on a single system Hamiltonian deformation.  For example for the $R$ system we would have
\begin{align}
\delta H(t) = \wt{\lambda} A(t)
\end{align}
where $\wt{\lambda} = \lambda \sqrt{\lan B^2(t_0) \ran}$. Since this amounts to acting with a unitary on either system we can make definite statements about how the energy will change. Since system $X$ begins in the ground state, this must increase the energy. The same conclusion would hold for $R$ in the thermal state, since this state minimizes the expectation value of the Hamiltonian while keeping fixed the entanglement entropy. 

The more careful argument is the following. The commutator can be written as 
\begin{align}
 i \lan  [\dot{A}(t_0), A(t_0)] \ran &= i \Tr\left[e^{- \beta H} \left( \dot{A}(t_0 - i \tau) A(t_0) - A(t_0) \dot{A}(t_0 + i \tau) \right)\right] \Big|_{\tau \rightarrow 0} \\
 &= - \partial_\tau \Tr\left[e^{- \beta H} \left( {A}(t_0 - i \tau) A(t_0) + A(t_0) {A}(t_0 + i \tau) \right)\right] \Big|_{\tau \rightarrow 0} \\ 
 &= - 2 \partial_\tau \Tr \left[ e^{- (\beta - \tau) H} A(t_0) e^{- \tau H} A(t_0)  \right] \Big|_{\tau \rightarrow 0} 
\end{align}
We want to show that 
\begin{align}
\partial_\tau  \Tr \left[ e^{- (\beta - \tau) H} A(t_0) e^{- \tau H} A(t_0)  \right]\Big|_{\tau \rightarrow 0} <0
\end{align}
This is not hard to prove. Consider working out the trace in the energy basis. This gives
\begin{align}
-\sum_{n m} |A_{nm}|^2 e^{- \beta E_n - \tau (E_m - E_n)} (E_m - E_n)
\end{align}
which after noting that $|A_{nm}|$ is symmetric in $n$ and $m$ can be re-expressed as
\begin{align}
-2 \sum_{n >m}|A_{nm}|^2 e^{\beta (E_n + E_m)/2} (E_m - E_n) \sinh\left[ \left({\beta \over 2} - \tau\right) \left( E_m - E_n\right) \right]
\end{align}
which is indeed negative for $\tau = 0$. This shows that 
\begin{align}
 i \lan  [\dot{A}(t_0), A(t_0)] \ran > 0
\end{align}
in the thermal state. The same conclusion would hold for a state that has thermalized and for a simple operator $A$.

\bibliographystyle{jhep}
\bibliography{bibliography}
\end{document}